\documentclass[10 pt, conference]{ieeeconf}

\IEEEoverridecommandlockouts                              % This command is only needed if
\usepackage{graphicx}
\usepackage{subcaption}
\usepackage{bbm}
\usepackage{upgreek}
\usepackage{bbold}
\usepackage{color}
\usepackage{graphics} % for pdf, bitmapped graphics files
\usepackage{needspace}

\input{mysymbol.sty}
\usepackage{theorem}
\usepackage{cite}
\usepackage{amsmath}
\usepackage{amssymb}
\usepackage{mathtools}
\usepackage{multirow}
\usepackage{mathabx}
\usepackage{url}
\usepackage{graphics, times, amsfonts}
\usepackage{tikz, epic,eepic}
\usetikzlibrary{shapes,arrows}
\usepackage{pgfplots}
\usepackage{color}
\usepackage{hyperref}
\usepackage[capitalise]{cleveref}
\usepackage{epstopdf}
\usepackage{epsfig, amsbsy}
\usepackage{latexsym}
\usepackage{amscd, verbatim}
\usepackage{multirow}
\usepackage{booktabs}

\usepackage{algorithm}
\usepackage{algorithmic}

\renewcommand{\baselinestretch}{0.97}\selectfont
\newtheorem{theorem}{Theorem}

\newtheorem{corollary}{Corollary}
\newtheorem{lemma}{Lemma}

{\itshape}{\rmfamily}
\newtheorem{assumption}{Assumption}

\def\forall{\text{for all\ }}

\newcommand{\aij}{\alpha_{ij}}
\newcommand{\bij}{\beta_{ij}}
\newcommand{\Ni}{\mathcal{N}_i}
\newcommand{\Nil}{\mathcal{N}_i^{\mathcal{L}}}
\newcommand{\Nim}{\mathcal{N}_i^{\mathcal{M}}}
\newcommand{\xl}{x_{\mathcal{L}}}
\newcommand{\Txl}{\Tilde{x}_{\mathcal{L}}}
\newcommand{\xm}{x_{\mathcal{M}}}
\newcommand{\wl}{W_{\mathcal{L}}}
\newcommand{\barwl}{\widebar{W}_{\mathcal{L}}}
\newcommand{\wm}{W_{\mathcal{M}}}

\newcommand{\an}[1]{{\textcolor{magenta}{#1}}}

\newcounter{sideremark}

\title{\LARGE 
Multi-Agent Resilient Consensus under Intermittent Faulty and Malicious Transmissions}
\author{Sarper Ayd\i n, Orhan Eren Akg\"un, Stephanie Gil, and Angelia Nedi\'c % 
\thanks{S. Ayd\i n and A. Nedi\'c are with the School of Electrical, Computer and Energy Engineering, Arizona State University, Tempe, AZ 85281. E-mail:{\tt\small \; saydin2@asu.edu; angelia.nedich@asu.edu.}  O. E. Akg\"un and S. Gil are with the School of Engineering and Applied Sciences, Harvard
University, Cambridge, MA 02139. E-mail: {\tt\small \;erenakgun@g.harvard.edu ; sgil@seas.harvard.edu}.
This work has been supported by the NSF awards CNS-2147641 and CNS-2147694.}}

\begin{document}
\normalsize
\maketitle

%%%%%%%%%%%%%%%%%%%%%%%%%%%%%%%%%%%%%%%%%%%%%%%%%%%%%%%%%%%%%%%%%%%%%%%%%%%
%%%   A   B   S   T   R   A   C   T   %%%%%%%%%%%%%%%%%%%%%%%%%%%%%%%%%%%%%
%%%%%%%%%%%%%%%%%%%%%%%%%%%%%%%%%%%%%%%%%%%%%%%%%%%%%%%%%%%%%%%%%%%%%%%%%%%
%
\begin{abstract}

In this work, we consider the consensus problem in which legitimate agents share their values over an undirected communication network in the presence of malicious or faulty agents. Different from the previous works, we characterize the conditions that generalize to several scenarios such as intermittent faulty or malicious transmissions, based on trust observations. As the standard trust aggregation approach based on a constant threshold fails to distinguish intermittent malicious/faulty activity, we propose a new detection algorithm utilizing time-varying thresholds and the random trust values available to legitimate agents. Under these conditions, legitimate agents almost surely determine their trusted neighborhood correctly with geometrically decaying misclassification probabilities. We further prove that the consensus process converges almost surely even in the presence of malicious agents. We also derive the probabilistic bounds on the deviation from the nominal consensus value that would have been achieved with no malicious agents in the system. Numerical results verify the convergence among agents and exemplify the deviation under different scenarios.
\end{abstract}

%%%%%%%%%%%%%%%%%%%%%%%%%%%%%%%%%%%%%%%%%%%%%%%%%%%%%%%%%%%%%%%%%%%%%%%%%%%
%%%   E N D     A   B   S   T   R   A   C   T   %%%%%%%%%%%%%%%%%%%%%%%%%%%%%%%%%%%%%
%%%%%%%%%%%%%%%%%%%%%%%%%%%%%%%%%%%%%%%%%%%%%%%%%%%%%%%%%%%%%%%%%%%%%%%%%%%
%

%%%%%%%%%%%%%%%%%%%%%%%%%%%%%%%%%%%%%%%%%%%%%%%%%%%%%%%%%%%%%%%%%%%%%%%%%%%
%%%   S E C T I O N %%%%%%%%%%%%%%%%%%%%%%%%%%%%%%%%%%%%%
%%%%%%%%%%%%%%%%%%%%%%%%%%%%%%%%%%%%%%%%%%%%%%%%%%%%%%%%%%%%%%%%%%%%%%%%%%%
%
\section{Introduction}

% !TEX root = dfp_com_aware.tex
%\textcolor{cyan}{
%In this paper we are interested in average consensus problem in multi-agent systems, where agents try to reach agreement on certain quantities using only local communication and computation. Consensus algorithms constitute a basis for distributed decision-making in networked multi-agent systems \cite{NedicOzdaglar,nedic2018distributed}, and are relevant for many multi-agent coordination applications, such as determining heading direction and velocity \cite{kia2019tutorial,martinez2009distributed}}

%Multi-agent systems \cite{dorri2018multi} are the new generation of technological structures, consisting of multiple entities having separate local information from each other. In the given setting, agents share their information via a communication network to collaborate and fulfil assigned tasks.

In this paper we are interested in the consensus problem \cite{degroot1974reaching,olfati2004consensus} in cyberphysical multi-agent systems under intermittent malicious attacks or failures. Agents need to reach an agreement over a set of variables using only local computation and communicating over a static undirected graph in the presence of malicious (non-cooperative) agents. Consensus algorithms constitute a basis for distributed decision-making in networked multi-agent systems \cite{NedicOzdaglar,nedic2018distributed}, and are relevant for many multi-agent coordination applications, such as determining heading direction, rendezvous, and velocity agreement\cite{kia2019tutorial,bullo2006rendezvous, martinez2009distributed}. However, consensus algorithms that assume all agents are cooperative are known to be susceptible to malicious and faulty behaviors \cite{pasqualetti2011consensus,sundaram2010distributed}. Our goal in this work is to develop a detection method that utilizes "side information" from the physical aspects of cyber-physical systems and the associated consensus algorithm, aiming to mitigate intermittent attacks and failures. A major difference to previous work is that here we treat the time-varying case, where agents behave maliciously intermittently.

Achieving resilient consensus in the presence of malicious agents has been studied extensively in the literature. Earlier methods that only use the transmitted data to detect or eliminate untrustworthy information impose restrictions on the connectivity of the network and the number of tolerable malicious agents \cite{dolev1982byzantine,pasqualetti2011consensus,sundaram2010distributed,leblanc2013resilient}. As these fundamental limitations apply to other distributed computation algorithms \cite{sundaram2019optimization,rawat2011sensing}, researchers have explored leveraging additional information, that can be obtained from the physicality of the system, to assess the trustworthiness of the agents~\cite{gil2017guaranteeing,cavorsi2023ICRA,Pierson2016,xiong2023securearray}. For example, in a Sybil attack where a malicious agent can spoof nonexistent entities to gain a disproportionate influence in the consensus \cite{gil2019consensus}, or a location misreporting attack in multi-robot location coverage \cite{gil2017guaranteeing}, work in \cite{gil2017guaranteeing} shows that stochastic trust observations $\alpha_{ij}(t)\in[0,1]$—indicating the trustworthiness of a link $(i,j)$ at time $t$—can be derived from wireless signal information. In these cases $\alpha_{ij}(t)$ captures properties of uniqueness (for the Sybil Attack), and/or uses cross-validation of the direction of arrival of the signal (for the location misreporting attack), and satisfies the property that $\mathbb{E}(\alpha_{ij}(t))\geq 1-\epsilon$ when $j$ is a legitimate agent and $\mathbb{E}(\alpha_{ij}(t))\leq \epsilon$ when $j$ is a malicious agent for some constant $\epsilon<1/2$~\cite{gil2017guaranteeing}. While these are specific examples, the concept of capturing the likelihood that an agent $j$ is potentially malicious via a stochastic observation of trust is a general one (see survey paper~\cite{gil2023physicality}).

Previous work \cite{yemini2021characterizing} shows that agents can detect untrustworthy agents over time using the trust observations and reach consensus even when malicious agents are in the majority. In particular, as the bounds on the expectations of trust observations $\epsilon$ is less than a known threshold $1/2$, legitimate agents can provably distinguish the trustworthiness of their neighbors by aggregating more observations over time and comparing the aggregated trust values to a predefined threshold. However, this ability breaks under \emph{intermittent} attacks of the malicious agents. In certain cases, malicious agents can inflict more damage to distributed systems by attacking randomly instead of attacking all the time \cite{nurellari2018detection,kailkhura2015detection}. Notably, this assumption is violated even for \emph{unintentional} behavior such as intermittent failures due to noisy sensors leading to incorrect location reporting. The reason is that intermittent attacks result in in a mixture of trustworthy and untrustworthy transmissions, precluding the ability to differentiate an attacker from a legitimate agent by using a constant threshold (for example based on $\epsilon<1/2$) as was the case in previous work \cite{yemini2021characterizing,akgun2023learning}. Standard statistical tests necessitate the knowledge and certain forms of the distributions where samples are drawn, e.g. their moments and continuity \cite{kay1998fundamentals,berger2014kolmogorov}. However, such properties may not be available to agents or may not hold with intermittent malicious transmissions, leading to the unavailability of convergence guarantees for the tests.

We address these challenges by proposing a new detection algorithm and a consensus method providing resilience against intermittent attacks and failures. Our detection method utilizes key observations that legitimate agents trust observations are sampled from the same distribution and that their expectations are higher than the malicious agents even when they act intermittently malicious. In the proposed algorithm, agents accumulate trust values from neighbors over time. Each round, they select their most trusted neighbor (the one with the highest aggregate trust value) as a reference and construct a trusted neighborhood by comparing other agents' aggregate trust values to that of the most trusted neighbor. Agents employ an adaptive threshold that grows over time, allowing them to exclude all malicious agents eventually, while still keeping their legitimate neighbors in their trusted neighborhood. Agents perform consensus updates using the values coming from their trusted neighbors only. Under the assumption that all legitimate agents have at least one legitimate neighbor, we demonstrate that the probability of agents misclassifying their neighbors decreases geometrically over time, resulting in a period after which no classification errors occur. Moreover, we show that the legitimate agents reach consensus almost surely, and their deviation from the consensus value is bounded. More specifically, our contributions are as follows: 
\begin{enumerate}
    \item We introduce a detection algorithm (\cref{alg_trust_neig}) that distinguishes between legitimate and malicious agents without relying on a predefined threshold.
    \item We show that misclassification probabilities using the detection method decreases geometrically over time (Lemmas \ref{lem_misp_legit}-\ref{lem_misp_mal}). Moreover, we show that there exists a random but finite time where thereafter no classification errors occur (\cref{lem_als_w}).
    \item We show that legitimate agents can reach consensus almost surely, regardless of the frequency of the attack (\cref{cor_con}). Moreover, for a given confidence level, we explicitly characterize the maximal deviation from the consensus value based on the properties of the trust values capturing the effect of the intermittent attacks, parameters determining the growth rate of the adaptive threshold, and number of legitimate and malicious agents (Theorem \ref{thm_dev}).
    \item We validate our approach in numerical studies under random attack and failure cases, showing that agents achieve convergence over time and learn the trustworthiness of their neighbors as predicted by the analysis. 
\end{enumerate}

\section{Consensus Dynamics with Failures and Attacks}\label{sec:model}
\subsection{Notation}
We use $|.|$ to denote absolute values of scalars and cardinalities of sets. We write $[.]_{i}$ and $[.]_{ij}$ for the $i^{th}$ entry of a vector and the $ij$-th entry of a matrix, respectively. We also extend the notation $|.|$ to matrices/vectors to define the element-wise absolute value of matrices/vectors, e.g., $[|A|]_{ij}=|[A]_{ij}|$. For matrices $A$ and $B$, we write $A>B$ (or $A\ge B$) when $[A]_{ij}>[B]_{ij}$ 
(or $[A]_{ij} \ge [B]_{ij}$) for all $i,j$. 
We use $\mathbf{0}$ and $\mathbf{1}$ to represent vectors/matrices
whose entries are all 0 and 1, respectively, without explicitly stating their dimensions as they can be understood within the context. 
We also use the backward matrix product of the matrices $H_k$, defined as follows:
\begin{align}
    \prod_{k=\tau}^t H_k, \begin{cases}
 H_t \cdots H_{\tau-1} H_{\tau} \: &\text{if } t \ge \tau, \\
I \: &\text{otherwise},
\end{cases}
\end{align}
where $I$ corresponds to the identity matrix.

%----------------------------------------------------------
\subsection{Consensus in Presence of Untrustworthy Agents}
%----------------------------------------------------------
We study the consensus dynamics among multiple agents defined by the set $\mathcal{N}:= \{1,\ldots, N\}$. The agents send and receive information through a static undirected graph  $G(\mathcal{N}, \mathcal{E})$, where $\mathcal{E}\subseteq \mathcal{N} \times \mathcal{N}$ represents the set of undirected edges among the agents. 
For each agents $i$, the set  of neighboring agents
is denoted by $\Ni:=\{ j \in \ccalN: (i,j) \in \ccalE\}$.
The agent set $\mathcal{N}$ consists of legitimate agents who are always trustworthy and malicious agents who can be trustworthy or not. The set of legitimate agents is denoted by $\mathcal{L}$, while the set of malicious agents is denote by $\mathcal{M}$, with $\mathcal{L}\cup\mathcal{M}=\mathcal{N}$ and $\mathcal{L}\cap \mathcal{M}=\emptyset$. These sets are fixed over time and assumed to be unknown.
The legitimate agents have associated nonnegative weights, subject to changes over time, for the existing communication links such that $w_{ij}(t) \in  [0,1]$  if $(i,j) \in \mathcal{E}$, otherwise $w_{ij}(t)=0$. 
The consensus dynamics among the agents starts at some time $T_0\ge0$, and we model the dynamic for the legitimate agents, as follows: for all $i \in \ccalL$ and $\forall t \ge T_0-1$, 
\begin{equation} \label{eq_con}
     x_i(t+1)=w_{ii}(t)x_i(t)+\sum_{j \in \mathcal{N}_i} w_{ij}(t)x_j(t),
 \end{equation}
 where $x_i(t) \in \reals$ for all $i \in \ccalL$.
 According to this update rule,
 each legitimate agent $i \in \ccalL$ takes a convex combination of its value and its neighbors, i.e. $w_{ii}(t) >0$, $w_{ij} (t)\ge 0$, and $w_{ii}(t)+ \sum_{j \in \mathcal{N}_i} w_{ij}(t)=1$.  Since the consensus update starts at time $T_0$, we assume that $x_i(0)=x_i(t)$ for all $0\le t < T_0$.
 The dynamic of the malicious agent's values is assumed to be unknown even in the case they are not actively attacking, and it is not modeled. %Therefore, this analysis captures a worst-case scenario.

We define $x(t) \in \reals^N$ as a vector of agents' values at time $t$. Given the partition of the agents as legitimate and malicious, 
 we partition the vector $x(t)$ accordingly, i.e., 
 $x(t)=[ \xl(t), \xm(t)]^T$. Then, the consensus dynamics~\eqref{eq_con} can be written in a vector notation:
 \begin{equation} \label{eq_con_dy}
\xl(t+1) 
=
     \begin{bmatrix}
\wl(t) & \wm(t)
\end{bmatrix}
.
  \begin{bmatrix}
\xl(t) \\
\xm(t) 
\end{bmatrix},
 \end{equation}
where $\wl(t) \in \mathbb{R}^{|\ccalL|\times |\ccalL|}$ and $\wm(t) \in \mathbb{R}^{|\ccalL|\times |\ccalM|}$ are the weight matrices associated with legitimate and malicious agents. 
 %Similarly, $\Theta(t)$ and $\Omega(t)$ are unknown dynamics providing the change in the values of malicious agents. Note that indices can be ordered to group legitimate and malicious agents inside the aforementioned vectors and weight matrices in \eqref{eq_con} holds without a loss of generality. 
Hence, the consensus dynamics of legitimate agents can be written as a sum of two terms at any time $t \ge T_0$, %starting from time $T_0$,
 \begin{equation}\label{eq_con_sum}
     \xl(T_0,t)= \Tilde{x}_{\ccalL}(T_0,t)+\phi_{\ccalM}(T_0,t),
 \end{equation}
%where the contributions of legitimate and malicious agents are defined as follows:
where
\begin{align}  
    \Tilde{x}_L(T_0,t)&= \bigg ( \prod_{k=T_0-1}^{t-1} \wl(k) \bigg)\xl(0), \label{eq_con_part1}\\
    \phi_{\ccalM}(T_0,t)&=  \sum_{k=T_0-1}^{t-1} \bigg ( \prod_{\ell=k+1}^{t-1} \wl(\ell) \bigg) \wm(k) \xm(k). \label{eq_sep_Dy}
\end{align}
Here, the term $\Tilde{x}_L(T_0,t)$ represents the influence of legitimate agents on each other and the term $\phi_{\ccalM}(T_0,t)$ represents the influence of malicious agents on the legitimate agents' values. These relations in~\eqref{eq_con_sum}-\eqref{eq_sep_Dy} are the backbone of the subsequent analysis, as they capture the consensus dynamics of the legitimate agents in terms of the starting time $T_0$, the initial values $x(0)$, together with the malicious inputs $\xm(k)$. 

We assume that the values $x_i(t)$ of all agents are bounded with a parameter $\eta>0$, i.e., $|x_i(t)| \le \eta$ for all $i\in\mathcal{N}$, and this parameter is known by all agents. Under this assumption,  no malicious agent will ever send a value that exceeds $\eta$ for otherwise it will be immediately detected. This assumption is crucial for bounding the cumulative impact of malicious inputs, as
captured by $\phi_{\ccalM}(T_0,t)$ in~\eqref{eq_sep_Dy}.

%----------------------------------------------
\subsection{Trusted Neighborhood Learning}
%-----------------------------------------------

Each legitimate agent $i \in \ccalL$  aims to classify its legitimate neighbors $\Nil := \Ni \cap \mathcal{L}$ and malicious neighbors $\Nim := \Ni \cap \mathcal{M}$ correctly over time by gathering trust values $\alpha_{ij}(t)$ for each transmission from their neighbors  
$j \in \Ni$ (see~\cite{gil2017guaranteeing} for more details on how to compute the trust values $\alpha_{ij}(t)$). 
The values $\alpha_{ij}(t)$, $t\ge0$, are random with values in the unit interval, i.e.,  $\alpha_{ij}(t) \in [0,1]$ for all $j\in\mathcal{N}$ and all $t\ge0$, 
where higher $\alpha_{ij}(t)$ values ($\alpha_{ij}(t) \rightarrow 1$) indicate the event that a neighbor $j$ is legitimate, is more likely.

The legitimate agents utilize the observed trust values $\{\alpha_{ij}(k)\}_{0\le k \le t}$ to determine 
their trustworthy neighbors and select the weights $w_{ij}(t)$ at time $t$. Following the work in~\cite{yemini2021characterizing}, we use the aggregate trust values, i.e.,
\begin{equation}
    \bij(t)=\sum_{k=0}^{t} (\aij(k)-1/2), \;  \forall \: i \in \ccalL\hbox{ and } j \in \Ni. 
\end{equation}
%Similarly, trust values may not be available to legitimate agents at all time steps due to the limited capabilities of legitimate agents, where they set $\alpha_{ij}(t)$ to $0$ for unavailable observations.
We make the following assumption on the trust values $\alpha_{ij}(t)$.

\begin{assumption} \label{as_trust}
Suppose that the following statements hold.
\begin{enumerate}
    \item[\textit{(i)}] The expected value of malicious and legitimate transmissions received by a legitimate agent $i \in \mathcal{L}$ are constant and satisfy
    \begin{align}
        c_j&=\mathbb{E}(\aij(t)), \: j \in \Nim, \\
        d&=\mathbb{E}(\aij(t)), \: j \in \Nil, 
    \end{align}
    where $d-c_j >0$ for all $j \in \Nim$.
    %where only the legitimate transmissions have identical expected values, \textit{i.e}, $d=d_j$ for all $j \in \Nil$.
    \item[\textit{(ii)}] The random variables $\alpha_{ij}(t)$ observed by a legitimate agent $i \in \mathcal{L}$ are independent and identically distributed for a given agent $j \in \mathcal{N}_i$ at any time index $t \in \mathbb{N}$.

    \item[\textit{(iii)}]  The subgraph $G_{\ccalL}=(\ccalL, \ccalE_{\ccalL})$ induced by the set of legitimate agents $\mathcal{L}$ is connected, where $\ccalE_{\ccalL}:= \{(i,j) \in \ccalE: (i,j) \in \ccalL \times \ccalL \}$.
\end{enumerate}
\end{assumption}

Assumption~\ref{as_trust}-\textit{(i)} captures the scenarios
where each malicious agent $m \in \mathcal{M}$ transmits malicious information with a (nonzero) probability $p_m \in (0,1]$ at each time $t$, whereas a legitimate agent $l \in \mathcal{L}$ never exhibits malicious behavior. It also holds when a malicious agent periodically sends malicious information in (deterministic) bounded time intervals. As a result, different and unknown rates of malicious transmissions correspond to mixed distributions with different expectations for the malicious agents' trust values. 
Assumption \ref{as_trust}-\textit{(ii)} requires independent trust samples of each neighbor from identical distributions; note that distributions of trust values given an agent $j \in \Ni$ are identical, while these distributions can be a mixture of several distributions. Assumption \ref{as_trust}-\textit{(iii)} imposes the connectivity among legitimate agents with a fixed topology.

Unlike the work in~\cite{yemini2021characterizing},
{\it we assume that the legitimate agents do not have any apriori threshold values to determine their trusted neighborhood}. This phenomenon can arise  due to the dynamic behaviour of the malicious agents. For example,  a malicious agent $m$ can have faulty and/or malicious transmissions with some probability $p_m\in (0,1]$, while it can send legitimate values with probability $1-p_m$.

\begin{comment}
Then, the expected trust difference between a legitimate transmission $l \in \Nil$and a malicious transmission $m \in \Nim$ becomes, 
\begin{align} \label{eq_exp_dif}
    &E(\alpha_{il}(t))-E(\alpha_{im}(t))\\
    &= p_a^2(d-((1-p_m) d+(p_m\Tilde{c}_m))+ (1-p_a)p_a(d-0)\\
    & +(1-p_a)p_a(0-((1-p_m) d+(p_m\Tilde{c}_m))\\
    &+(1-p_a)^2(0-0)= (2p_a-p_a^2) p_m(d-\Tilde{c}_m)\ge \lambda>0 \label{eq_exp_dif_last}.
\end{align}
where $\Tilde{c}_m$ is the expected value of pure malicious behavior ($p_m=1$) satisfying $d-\Tilde{c}_m\ge \Tilde{\lambda}>0$. as shown in the example (Eqs. \eqref{eq_exp_dif}-\eqref{eq_exp_dif_last}).
\end{comment}

To handle the situation when an apriori threshold is unavailable, we propose a new learning method that legitimate agents implement to identify their trustworthy neighbors over time.
%We propose an algorithm to let agents determine their trusted neighborhood $\Ni(t)$ accurately, i.e., $\Ni(t)=\Nil$. 
The algorithm is built on three properties,
(1)~all legitimate agents have at least one legitimate neighbor, i.e., $\Nil \neq \emptyset$, (Assumption~\ref{as_trust}-\textit{(iii)}) 
(2)~the legitimate agents have identical aggregate trust values in expectation, and (3)~the legitimate agents have higher trust values compared to malicious agents in expectation (see Assumption~\ref{as_trust}-\textit{(i)}). 
Based on property (3), in the algorithm, each legitimate agent chooses the highest aggregate trust value as a reference point. Then, it eliminates the malicious agents based on the unbounded (expected) difference of trust value aggregates between a legitimate and a malicious agent, as $t \xrightarrow{} \infty$ (based on Assumption \ref{as_trust}-\textit{(i)}). Algorithm \ref{alg_trust} is provided as below.

%{\bf\color{red} Sarper, in algorithm 1, we need to add that every legimate agent $i$ performs steps 2-4.}\sa{I transformed them into sentences with "each agent $i$".}
\begin{algorithm}[H] 
   \caption{Trusted Neighborhood Learning}
\label{alg_trust}
\begin{algorithmic}[1]\label{alg_trust_neig}

  \STATE {\bfseries Input:} Threshold value $\xi >0$, $\gamma \in (0.5,1)$.
\STATE Each agent $i \in \ccalL$ finds $\bar{j}(t)=\arg \max_{j \in \mathcal{N}_i} \beta_{ij}(t)$.
\STATE  Each agent $i \in \ccalL$ checks if $ \beta_{i\bar{j}(t)}(t)- \beta_{ij}(t) \le \xi_t$, where $\xi_t=\xi (t+1)^{\gamma}$, $ \forall j \in \Ni$.
\STATE Each agent $i \in \ccalL$ returns $\Ni(t)=\{j \in \Ni| \beta_{i\bar{j}(t)}(t)- \beta_{ij}(t) \le \xi_t \}$.
   \end{algorithmic}
\end{algorithm}
\vspace{-3 mm}
In words, each legitimate agent $i\in\mathcal{L}$ first selects its most trusted neighbor (Step $2$). Then, it compares others' trusted values with the most trusted agent (Step $3$), and finally determines its trusted neighborhood with time-varying threshold values (Step $4$). The chosen range for $\gamma$ ensures the threshold grows slow enough to exclude malicious agents while maintaining a pace that retains legitimate agents over time. The rationale behind this selection will become more evident in \cref{lem_misp_legit} and \cref{lem_misp_mal} later on.
%{\bf\color{red} Sarper, It may be good to say that here we use time-varying threshold value ......}

Next, we define the actual weights $w_{ij}(t)= [W(t)]_{ij}$ assigned by legitimate agents $i \in \ccalL$ based on their 
learned trusted neighborhoods $\Ni(t)$, as follows:
\begin{equation} \label{eq_wij}
    w_{ij}(t)= \begin{cases}
 \frac{1}{n_{w_i}(t)} \: &\text{if } j \in \Ni(t), \\
1-\sum_{\ell \in \Ni (t)} w_{i\ell}(t) \: &\text{if } j=i,\\
0 \: &\text{otherwise},
\end{cases}
\end{equation}
%{\bf\color{red} Sarper, what is $\kappa$?} 
where $n_{w_i}(t)= \max\{|\Ni (t)|+1, \kappa \} \ge 1$. 
and $\kappa >0$ is a parameter that limits the influence of other agents on the values $x_i(t)$.
Similarly, we define the matrix $\widebar{W}_{\ccalL}$ 
that would have been constructed if the legitimate agents have known their trusted neighbors, i.e.,
for the pairs of agents $(i,j) \in \ccalL \times \mathcal{N}$, 
\begin{equation} \label{eq_barwij}
    [\widebar{W}_{\ccalL}]_{ij}= \begin{cases}
\frac{1}{\max\{ |\Nil|+1, \kappa\}} \: &\text{if } j \in \Nil,\\
1-\frac{1}{\max\{ |\Nil|+1, \kappa\}} \: &\text{if } j=i,\\
0 \: &\text{otherwise}.
\end{cases}
\end{equation}

The (nominal) matrix $\widebar{W}_{\ccalL}$ can be considered as an ideal and target case for each legitimate agent to eliminate the effect of malicious agents in the consensus process defined in Eqs. \eqref{eq_con_sum}-\eqref{eq_sep_Dy}.

%%%%%%%%%%%%%%%%%%%%%%%%%%%%%%%%%%%%%%%%%%%%%%%%%%%%%%%%%%%%%%%%%%%%%%%%%%%
%%%   S E C T I O N %%%%%%%%%%%%%%%%%%%%%%%%%%%%%%%%%%%%%
%%%%%%%%%%%%%%%%%%%%%%%%%%%%%%%%%%%%%%%%%%%%%%%%%%%%%%%%%%%%%%%%%%%%%%%%%%%
%

\section{Analysis of Consensus Dynamics} \label{sec::conv}
%In this section, we show the convergence of $\wl (t)$ to $\widebar{W}_{\ccalL}$

%----------------------------------------------
\subsection{Convergence of Consensus Dynamics}
%----------------------------------------------

We start by analyzing the probability that a legitimate agent $i$ misclassifies one of its neighbors at some time $t$ in \cref{alg_trust_neig}. This misclassification can occur in two ways. A legitimate agent $i$ can misclassify one of its legitimate neighbors $j \in \Nil$ as malicious, resulting in agent $j$ being excluded from the trusted neighborhood $\Ni(t)$. Conversely, agent $i$ can misclassify one of its malicious neighbors $m \in \Nim$ as legitimate, resulting in agent $m$ being included in the trusted neighborhood $\Ni(t)$. \cref{fig:alg_intuition} shows examples of how these misclassifications can occur.

\begin{figure}
     \centering
     \begin{subfigure}[b]{0.4\textwidth}
         \centering
         \includegraphics[width=\textwidth]{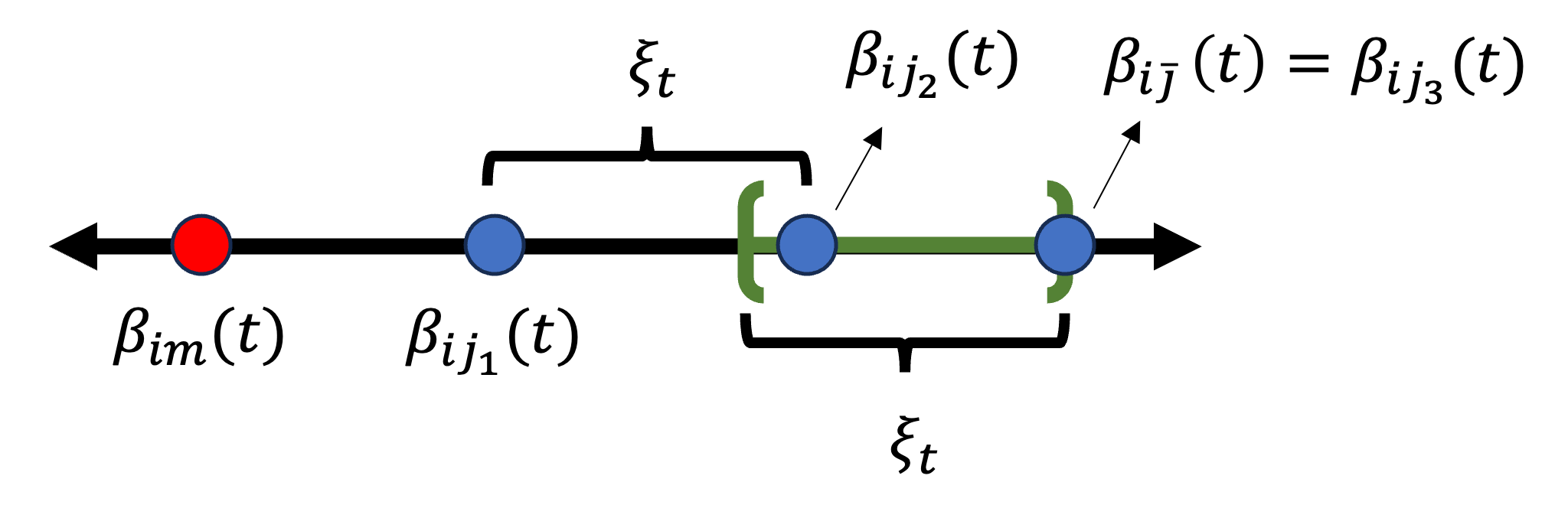}
         \caption{A legitimate neighbor is misclassified as malicious}
         \label{fig:alg_intuition_a}
     \end{subfigure}
     \hfill
     \begin{subfigure}[b]{0.4\textwidth}
         \centering
         \includegraphics[width=\textwidth]{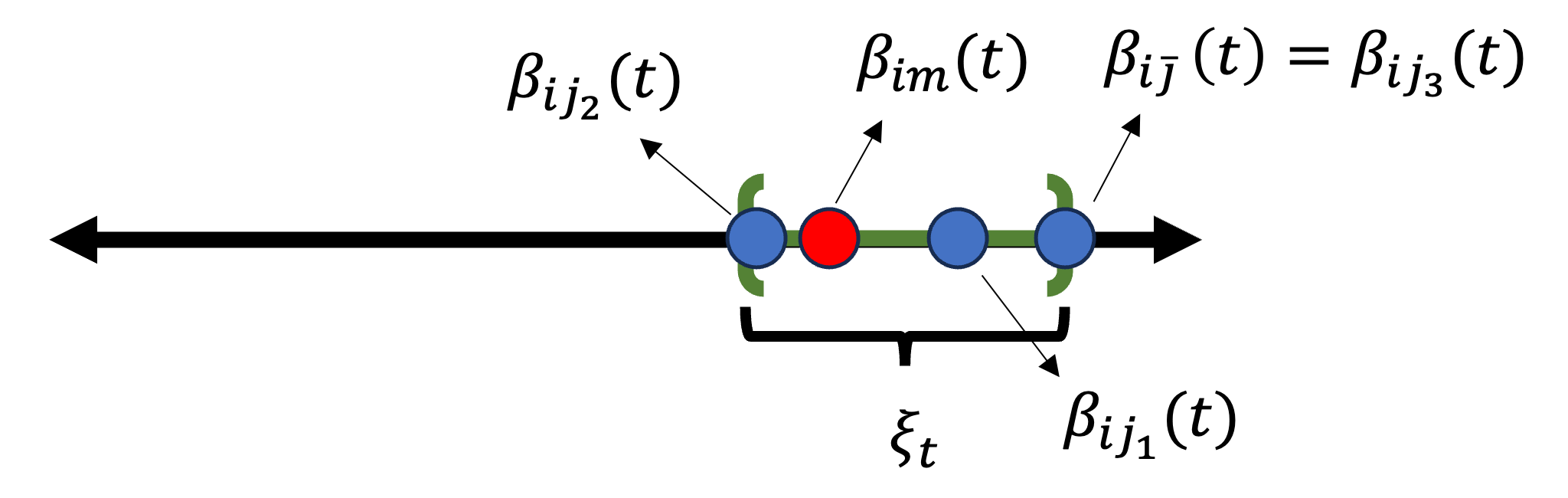}
         \caption{A malicious neighbor is misclassified as legitimate}
         \label{fig:alg_intuition_b}
     \end{subfigure}
    \caption{Misclassification examples for an agent $i$ with three legitimate neighbors with aggregate trust values $\beta_{ij_1}(t)$, $\beta_{ij_2}(t)$, $\beta_{ij_3}(t)$ and a malicious neighbor with aggregate trust value $\beta_{im}(t)$. The aggregate trust values are placed on a number line where values to the right are larger and values to the left are smaller. The green area inside the brackets indicates the trusted region $\xi_t$ from Algorithm~\ref{alg_trust_neig}.} %(a) We have misclassification of the legitimate neighbor $j_1$. The fact that $\beta_{ij_1}(t)$ is on the left of $\beta_{ij_2}(t)$ by $\xi_t$ implies that its gap to $\beta_{i\bar{j}}(t)$ is also at least $\xi_t$. Conversely, since it is more than $\xi_t$ to the left of $\beta_{i\bar{j}}(t)$, there must exist an agent to the right of $\beta_{ij_1}(t)$ by at least $\xi_t$. In this example, both $j_2$ and $j_3$ satisfy this implication. (b) We have misclassification of the malicious neighbor $m$. Since $\beta_{im}(t)$ lies within $\xi_t$ distance of $\beta_{i\bar{j}}(t)$, the remaining aggregate trust values can either lay on its left or be less than $\xi_t$ distance on its right. Conversely, since all the aggregate trust values are either on the left of $\beta_{im}(t)$ or within $\xi_t$ distance to its right, $\beta_{i\bar{j}}(t)$ cannot be more than $\xi_t$ distance to its right either.}
    
    \label{fig:alg_intuition}
    \vspace{-3 mm}
\end{figure}
An agent $i$ misclassifies a legitimate neighbor $j$ if the difference between the maximum aggregate trust value $\beta_{i\bar{j}(t)}(t)$ and $j$'s aggregate trust value $\beta_{ij}(t)$ exceeds $\xi_t$, i.e., $\beta_{i\bar{j}(t)}(t) - \beta_{ij}(t) > \xi_t$ (refer to Step 4 in \cref{alg_trust_neig}). This also means that there exist an agent $j'\in \Ni$ such that $\beta_{ij'}(t)-\beta_{ij}(t) > \xi_t$. This is true since if such an agent exists, then difference to $\beta_{i\bar{j}(t)}(t)$ can be only larger, knowing that $\beta_{i\bar{j}(t)}(t)$ is the maximum value. This observation is key to  the proof of \cref{lem_misp_legit}. \cref{fig:alg_intuition} (a) provides a graphical intuition on this. If agent $j'$ is legitimate, the probability of the event $\beta_{ij'}(t)-\beta_{ij}(t) > \xi_t$ decreases as agent $i$ aggregates more trust values, since trust values of $j$ and $j'$ come from the same distribution. If agent $j'$ is malicious, again, the probability of the event $\beta_{ij'}(t)-\beta_{ij}(t) > \xi_t$ decreases over time, since the malicious agents trust values are lower in expectation. The following lemma builds on this intuition to show that the probability of misclassifying a legitimate neighbor decreases exponentially over time.
\vspace{-1 mm}
\begin{lemma} \label{lem_misp_legit}
 Suppose Assumption \ref{as_trust} holds. Let $\xi>0$ and $\gamma \in (0.5,1)$ be the parameters defined in Algorithm \ref{alg_trust}. Let $j$ be an arbitrary legitimate neighbor of a legitimate agent $i$, i.e., $j \in \Nil$ for agent $i \in \mathcal{L}$. Then, the misclassification probability of for agent $j$ by agent $i$ has the following upper bound:
\begin{align*}
    \mathbb{P}(j \not \in &\Ni(t)) \leq |\Nil|\exp(-\xi^2 (t+1)^{2\gamma}/2(t+1))\\
     &+|\Nim|\exp(-(\xi (t+1)^{\gamma}+\lambda(t+1))^2/2(t+1)),
    %&\sa{= O( \exp( -\Tilde{\xi}^2(t+1)^{\min((2\gamma-1),\gamma)})),} 
\end{align*}
where $\lambda = \min_{m \in \mathcal{M}} (d-c_m)$.
\end{lemma}
\begin{proof}
    By \cref{alg_trust_neig}, a legitimate neighbor $j$ is mislassified when the event $\{\beta_{i\bar{j}(t)}(t)-\beta_{ij} (t) > \xi_t\}$ occurs. First, notice that this event is equivalent to the union of the events $\bigcup_{n \in \Ni} 
        \{ \beta_{in}(t)-\beta_{ij}(t) > \xi_t \}$. We check this equivalence in both directions. The event $\beta_{i\bar{j}(t)}(t) > \beta_{ij} (t) + \xi_t$ implies that there exist an element $j'\in \Ni$ such that $\beta_{ij'}(t) > \beta_{ij} (t) + \xi_t$, as we can simply choose $j'=\bar{j}(t)$. The converse is also true. The existence of a $j'\in \Ni$ with $\beta_{ij'}(t) > \beta_{ij} (t) + \xi_t$ implies that $\beta_{i\bar{j}(t)}(t) > \beta_{ij} (t) + \xi_t$ since $\beta_{i\bar{j}(t)}(t) \geq \beta_{ij'}(t)$ by definition given in \cref{alg_trust_neig} Step 2. Therefore, we have that the event $\{\beta_{i\bar{j}(t)}(t)-\beta_{ij} (t) > \xi_t\}$ is equivalent to the union of events $  \bigcup_{n \in \Ni} 
        \{ \beta_{in}(t)-\beta_{ij}(t) > \xi_t \}.$ Using this equality, we have,
    \begin{align}
        \mathbb{P}(j \not \in \Ni(t)) 
        = & \mathbb{P}(\beta_{i\bar{j}(t)}(t)-\beta_{ij} (t) > \xi_t) \\
        = & 
        \mathbb{P}(\bigcup_{n \in \Ni} 
        \{ \beta_{in}(t)-\beta_{ij}(t) > \xi_t \}) \\
        \leq & 
        \sum_{l \in \Nil \setminus\{j\}} \mathbb{P}( \beta_{il}(t)-\beta_{ij}(t) > \xi_t) \\
        &+ \sum_{m \in \Nim} \mathbb{P}( \beta_{im}(t)-\beta_{ij}(t) > \xi_t),
        \label{eq_j_legitimate_union_bound}
    \end{align}
    \vspace{-0.7 mm}
where the last step follows from the union bound. First, we focus on bounding the probability 
$\mathbb{P}( \beta_{il}(t)-\beta_{ij}(t) \ge \xi_t)$. Notice that $\beta_{il}(t)-\beta_{ij}(t) = \sum_{s=0}^t (\alpha_{il}(s)-\alpha_{ij}(s))$ is the sum of independent random variables $(\alpha_{il}(s)-\alpha_{ij}(s))$ with expectation $\mathbb{E}(\alpha_{il}(s)-\alpha_{ij}(s))=0.$ Therefore, we directly apply the Chernoff-Hoeffding inequality to obtain
    \begin{align*}
        \mathbb{P}( \beta_{il}(t)-\beta_{ij}(t) > \xi_t) &\leq \exp(-\xi_t^2/2(t+1)) \\
        &= \exp(-\xi^2 (t+1)^{2\gamma}/2(t+1)),
    \end{align*}
where in the last step we used the definition of $\xi_t$. Using a similar line of reasoning, 
we bound the probability $\mathbb{P}( \beta_{im}(t)-\beta_{ij}(t) > \xi_t)$, as follows:
\begin{align*}
        &\mathbb{P}( \beta_{im}(t)-\beta_{ij}(t) > \xi_t) \\
        %&= \mathbb{P}( \beta_{im}(t)-\beta_{ij}(t) - \mathbb{E}(\beta_{im}(t)-\beta_{ij}(t)) \nonumber \\
        %&> \xi_t - \mathbb{E}(\beta_{im}(t)-\beta_{ij}(t))) \\
        &\overset{(a)}{=} \mathbb{P}( \beta_{im}(t)-\beta_{ij}(t) - \mathbb{E}(\beta_{im}(t)-\beta_{ij}(t)) \nonumber\\
        &\hspace{15 mm}> \xi (t+1)^{\gamma} + (t+1)(d-c_m)) \\
        &\overset{(b)}{\leq} \exp(-(\xi (t+1)^{\gamma} + (t+1)(d-c_m))^2/2(t+1)),
    \end{align*}
where in $(a)$ we embed the expected difference of trust values into both sides, and in $(b)$ we apply the Chernoff-Hoeffding inequality since $(d-c_m)>0$ by \cref{as_trust}-\textit{(i)}. The rest of the proof follows from combining the bounds with \cref{eq_j_legitimate_union_bound}.
\end{proof}
\vspace{-0.5 mm}
Next, we analyze the probability of misclassifying a malicious neighbor $m \in \Ni$. Such misclassification happens if the gap between the maximum aggregate trust value $\beta_{i\bar{j}(t)}(t)$ and $m$'s value $\beta_{im}(t)$ is at most $\xi_t$, i.e., $\beta_{i\bar{j}(t)}(t) - \beta_{im}(t) \leq \xi_t$. The maximum value not being large enough to exclude agent $m$ from the neighborhood implies that no neighboring agents' aggregate trust values surpass $\beta_{im}(t)$ by $\xi_t$ (see \cref{fig:alg_intuition} (b) for a graphical intuition). Given that at least one legitimate neighbor exists for any agent $i$ as per \cref{as_trust}-\textit{(iii)}, the probability of misclassification is upper bounded by the probability that a legitimate neighbor's trust value does not exceed $\beta_{im}(t)$ by more than $\xi_t$. This rationale underpins the proof of the following lemma showing that the misclassfication probability decreases exponentially over time.
\vspace{-5 mm}
\begin{lemma} \label{lem_misp_mal}
  Suppose Assumption \ref{as_trust} holds. Let $\xi>0$ and $\gamma \in (0.5,1)$ be the parameters defined in Algorithm \ref{alg_trust}. Let $m$ be an arbitrary malicious neighbor of a legitimate agent $i$, i.e., $m \in \Nim$ for agent $i \in \mathcal{L}$. Then, for all $t > \left(\frac{\xi}{d-c_m}\right)^{1/(1-\gamma)},$ the misclassification probability of agent $m$ by agent $i$ has the following upper bound:
 \begin{align*}
     \mathbb{P}(m\in & \Ni(t)) %\nonumber\\
     \leq \exp(-(\xi (t+1)^{\gamma} + (t+1)\lambda)^2/2(t+1)),
     %& \sa{\le \exp(-(\xi^2/2)(d-c_m)^2 (t+1))}.
 \end{align*}
 where $\lambda = \min_{m \in \mathcal{M}} (d-c_m)$.
 %$\lambda$ is as defined in Lemma \ref{lem_misp_legit}.
\end{lemma}
\begin{proof}
    By the definition of the trusted neighborhood in Algorithm \ref{alg_trust}, a malicious neighbor $m$ is misclassified when we have  
    $\beta_{i\bar{j}(t)}(t)-\beta_{im}(t) \leq \xi_t.$ We note that this event is equivalent to the intersection of events $ \bigcap_{j \in \Ni} \{ \beta_{ij}(t)-\beta_{im}(t) \leq \xi_t \}$. This equivalence holds as we validate it from both directions. The event that the maximum element $\beta_{i\bar{j}(t)}(t)\leq \beta_{im}(t)+\xi_t$, which holds by the definition of $\Ni(t)$ (from Alg.~\ref{alg_trust_neig}) and $\beta_{i\bar{j}(t)}(t)$, implies that $\beta_{ij}(t)\leq \beta_{im}(t)+\xi_t$ for all $j\in \Ni$ (see Fig. 1.b). Conversely, if $\beta_{ij}(t)\leq \beta_{im}(t)+\xi_t$ for all $j\in \Ni$, then we have $\beta_{i\bar{j}(t)}(t)\leq \beta_{ij}(t)+\xi_t$ since $\bar{j}(t)$ is also chosen from the set of all $j\in \Ni$ (see Step 2 in \cref{alg_trust_neig}). Thus we have that the event $\{\beta_{i\bar{j}(t)}(t)-\beta_{im}(t) \leq \xi_t\}$ is equivalent to the intersection of events $  \bigcap_{j \in \Ni} \{ \beta_{ij}(t)-\beta_{im}(t) \leq \xi_t \}$. Using this equality, we get
    \begin{align*}
        \mathbb{P}(m\in \Ni(t)) &= \mathbb{P}(\beta_{i\bar{j}(t)}(t)-\beta_{im}(t) \leq \xi_t) \\
        &= \mathbb{P}(\bigcap_{n \in \Ni} \{ \beta_{in}(t)-\beta_{im}(t) \leq \xi_t \}) \\
        &\leq \min_{n\in \Ni}\mathbb{P}( \beta_{in}(t)-\beta_{im}(t) \leq \xi_t).
    \end{align*}
    Consider an arbitrary legitimate neighbor of an agent $i$, $l \in \Nil$. We know that such a neighbor must exist due to \cref{as_trust}-\textit{(iii)}. Then, we have
\begin{align*}
    &\min_{n\in \Ni}\mathbb{P}( \beta_{in}(t)-\beta_{im}(t) \leq \xi_t) \leq \mathbb{P}( \beta_{il}(t)-\beta_{im}(t) \leq \xi_t) \\
    %&= \mathbb{P}( \beta_{il}(t)-\beta_{ij}(t) - \mathbb{E}(\beta_{il}(t)-\beta_{ij}(t)) \leq \xi_t - \mathbb{E}(\beta_{il}(t)-\beta_{ij}(t))) \\
    &= \mathbb{P}( \beta_{il}(t)-\beta_{im}(t) - \mathbb{E}(\beta_{il}(t)-\beta_{im}(t)) \nonumber \\
    &\hspace{15 mm} \leq \xi (t+1)^{\gamma} - (t+1) (d-c_m)).
    \end{align*}
    Then, for $t > \left(\frac{\xi}{d-c_m}\right)^{1/(1-\gamma)},$ we have $\xi (t+1)^{\gamma} - (t+1)(d-c_m)<0$. Therefore, we apply the Chernoff-Hoeffding inequality to obtain the desired result. 
\end{proof}
Lemmas \ref{lem_misp_legit} and \ref{lem_misp_mal} show the misclassification probabilities go to $0$, as $t \rightarrow \infty$, due to the Chernoff-Hoeffding bound. The result stems from the fact that the probability of deviation of random sums from their expectations can be bounded. Since the misclassification probabilities decrease geometrically over time, the misclassfications will occur only finitely many times. Therefore, legitimate agents will only assign weights to their legitimate neighbors after some finite time, i.e., the weight matrices will be equal to the nominal weight matrix. The next lemma states this result formally.
%The next lemma states the convergence of weight matrices to the nominal weight matrix.
%The sequence of matrices $ \{\wl (t)\}_{t\ge T_0}$ converges to the matrix $\widebar{W}_{\ccalL}$ almost surely.
 \vspace{-2.5 mm}
\begin{lemma} \label{lem_als_w}
    Suppose Assumption \ref{as_trust} holds. 
    Let $\xi>0$ and $\gamma \in (0.5,1)$ be the parameters defined in Algorithm \ref{alg_trust}.There exists a (random) finite time $T_f >0$ such that $W_{\ccalL} (t)= \widebar{W}_{\ccalL}$ for all $t \ge T_f$. Furthermore, it holds almost surely
    \begin{equation}
        \prod_{t=T_0-1}^\infty \wl (t)= \mathbf{1}\nu^T \bigg (  \prod_{t=T_0-1}^{\max \{T_f,T_0\}-1} \wl (t) \bigg ),
    \end{equation}
    where the matrix product $\prod_{t=T_0-1}^\infty \wl (t) >{\bf 0}$ for any $T_0\ge0$ almost surely, and $\nu >{\bf 0}$ is a stochastic vector. 
\end{lemma}
\begin{proof}
     Using Lemmas \ref{lem_misp_legit} and \ref{lem_misp_mal}, legitimate agents have geometrically decaying misclassification probabilities.
    The infinite sums of misclassification probabilities satisfy $\sum_{t=0}^\infty \mathbb{P}( j \not \in \Ni (t))= \sum_{t=0}^\infty  O(\exp(-\xi^2 (t+1)^{2\gamma}/2(t+1)) < \infty $ for legitimate neighbors $ j \in \Nil$, and $\sum_{t=0}^\infty  \mathbb{P}( j  \in \Ni (t))=\sum_{t=0}^{T'-1} \mathbb{P}( j  \in \Ni (t))+\sum_{t=T'}^{\infty} O(\exp(-(\xi (t+1)^{\gamma} + (t+1)\lambda)^2/2(t+1))) < \infty$ 
    for all $t \ge T'$,
    where $ T' \ge \left(\frac{\xi}{d-c_m}\right)^{1/(1-\gamma)}$.  %\oea{OEA: I think the following sentence was written for the previous statement of the lemma. Maybe we should say misclassifications occur only finitely often by Borel-Cantelli lemma.} Then, we have the almost sure convergence of the matrices $\{\wl (t)\}_{t\ge T_0}$ as the result of the Borel-Cantelli lemma.
   Hence, there exists a finite time $T_f$ such that we have $W_{\ccalL} (t)= \widebar{W}_{\ccalL}$ for all $t \ge T_f$.
    %{\bf\color{red} Is the existence of the time $T_f$ also a consequence of that lemma; if so we need to say it.} 

    As $\wl (t)= \widebar{W}_{\ccalL}$ for all $t \ge T_f$, for  the product of the matrices $\wl(t)$, we have,
\vspace{-1 mm}
    \begin{align*}
        \prod_{t=T_0-1}^\infty \wl (t)&= \prod_{t=\max\{T_0,T_f\}}^\infty \wl (t) \prod^{\max\{T_0,T_f\}-1}_{t=T_0-1} \wl (t)\\
        &= \prod_{t=\max\{T_0,T_f\}}^\infty \widebar{W}_{\ccalL}  \prod^{\max\{T_0,T_f\}-1}_{t=T_0-1} \wl (t)\\
        &= \lim_{t \rightarrow \infty} \widebar{W}^{t-\max\{T_0,T_f\}}_{\ccalL} \prod^{\max\{T_0,T_f\}-1}_{t=T_0-1} \wl (t).
    \end{align*}
By Assumption \ref{as_trust}, the subgraph induced by the legitimate agents is connected and, by the definition of $\widebar{W}_{\ccalL}$, it follows that $\widebar{W}_{\ccalL}$ implies is a primitive stochastic matrix. Therefore,by the Perron-Frobenius Theorem, we have that $\lim_{t \rightarrow \infty} \widebar{W}^{t-\max\{T_0,T_f\}}_{\ccalL} = 1 \nu^T$, with $\nu>{\bf 0}$, and
\begin{equation*}
    \prod_{t=T_0-1}^\infty \wl (t)= \mathbf{1}\nu^T \bigg (  \prod_{t=T_0-1}^{\max \{T_f,T_0\}-1} \wl (t) \bigg ).
\end{equation*}
In addition, note that as $\nu$ is a stochastic and that 
diagonal entries of $\wl(t)$ are positive. Thus, $\prod_{t=T_0-1}^\infty \wl (t) >\mathbf{0}$  almost surely and $\nu^T\bigg (  \prod_{t=T_0-1}^{\max \{T_f,T_0\}-1} \wl (t) \bigg ) >\mathbf{0}$.
\end{proof}

We state the result on the consensus dynamics defined in \eqref{eq_con_sum} using the convergence of weight matrices.

\begin{lemma} \label{lem_con_legit}
    Suppose Assumption \ref{as_trust} holds. 
    In Algorithm \ref{alg_trust}, let $\xi>0$ and $\gamma \in (0.5,1)$. Let $\xl(0)$ be the initial values of legitimate agents. Then, $\Tilde{x}_{\ccalL}(T_0,t)$ converges almost surely, i.e., almost surely
    \begin{equation*}
            \lim_{t \rightarrow \infty}  \Tilde{x}_{\ccalL}(T_0,t)= \bigg ( \prod_{k=T_0-1}^{\infty} \wl(k) \bigg)\xl(0) = y\mathbf{1},
    \end{equation*}
where $y \in \mathbb{R}$ is a random variable depending on $T_f$ and $T_0$.
%being an element of the convex hull of initial values $\xl(0)$.
\end{lemma}
\begin{proof}
    The proof follows along the lines of Proposition 2 in \cite{yemini2021characterizing}.
\end{proof}

\begin{comment}

\begin{proof}
    The result follows from the almost sure convergence of $\{\wl(t)\}_{t\ge T_0}$ to $\widebar{W}_{\ccalL}$ established in Lemma \ref{lem_als_w}. Specifically, we have
    \begin{align*}
        &\bigg ( \prod_{k=T_0-1}^{\infty} \wl(k) \bigg)\xl(0)\\
        &= \bigg ( \prod_{k=\max\{T_0,T_f\}}^\infty \wl (k) \prod^{\max\{T_0,T_f\}-1}_{k=T_0-1} \wl (k) \bigg ) \xl(0)\\
        &= \bigg ( \prod_{k=\max\{T_0,T_f\}}^\infty \widebar{W}_{\ccalL}  \prod^{\max\{T_0,T_f\}-1}_{k=T_0-1} \wl (t) \bigg ) \xl(0) \\
        &= \mathbf{1} \nu^T \bigg ( \prod^{\max\{T_0,T_f\}-1}_{k=T_0-1} \wl (k) \bigg ) \xl(0),
    \end{align*}
    \an{and the desired relation follows by
    letting $y=\nu^T ( \prod^{\max\{T_0,T_f\}-1}_{k=T_0-1} \wl (k) ) \xl(0)$}.
\end{proof}
\end{comment}
%\an{Lemma~\ref{lem_con_legit} implies that} the elements of the vector  $ \Tilde{x}^\infty_{\ccalL}(T_0):=\lim_{t \rightarrow \infty}  \Tilde{x}_{\ccalL}(T_0,t)$ are equal to each other, $[\Tilde{x}^\infty_{\ccalL}(T_0)]_i=[\Tilde{x}^\infty_{\ccalL}(T_0)]_j$ for any $(i,j) \in \mathcal{L} \times \mathcal{L}$. 

We next state the limit of $\phi_{\ccalM}(T_0,t)$ (see~\eqref{eq_sep_Dy}) in the consensus process.

\begin{lemma} \label{lem_con_mal}
    Suppose Assumption \ref{as_trust} holds. 
    In Algorithm \ref{alg_trust}, let $\xi>0$ and $\gamma \in (0.5,1)$. Then, $\phi_{\ccalM}(T_0,t)$ converges almost surely, i.e., we have almost surely
    \begin{align*}
            \lim_{t \rightarrow \infty}\phi_{\ccalM}(T_0,t)&=  \sum_{k=T_0-1}^{\infty} \bigg ( \prod_{\ell=k+1}^{\infty} \wl(\ell) \bigg) \wm(k) \xm(k)\\
            &= f\mathbf{1},
    \end{align*}
where $f \in \mathbb{R}$ is a random variable depending on $T_f$ and $T_0$.
\end{lemma}
\begin{proof}
  The proof follows from Proposition 3 in \cite{yemini2021characterizing}, using the almost sure convergence of $W_{\ccalL}(t)$ and  $\wm (t)$.
\end{proof}

We now state the final result of this section concluding that legitimate agents reach a common value asymptotically,   
\begin{corollary} \label{cor_con}
    Suppose Assumption~\ref{as_trust} holds, and let $\xi>0$ and $\gamma \in (0.5,1)$ in Algorithm~\ref{alg_trust}. Then, the consensus protocol~\eqref{eq_con_dy} among the legitimate agents  converges almost surely, i.e.,
    \begin{equation}
         \lim_{t \rightarrow \infty}  \xl(T_0,t)=  z\mathbf{1}\quad\hbox{almost surely},
    \end{equation}
where $z \in \mathbb{R}$ is a random variable given by $z=y+f$, with $y$ and $f$ from Lemma~\ref{lem_con_legit} and Lemma~\ref{lem_con_mal}, respectively.
\end{corollary}
 
\begin{proof}
    The result follows by the relation $\xl(T_0,t)=\Tilde{x}_{\ccalL}(T_0,t)+\phi_{\ccalM}(T_0,t)$ (see~\eqref{eq_con_sum}) and Lemmas \ref{lem_con_legit}--\ref{lem_con_mal}.
    %we conclude that $\lim_{t \rightarrow \infty}  \xl(T_0,t)= \mathbf{1}y+\mathbf{1}f =\mathbf{1}z$.
\end{proof}

Corollary~\ref{cor_con} states that the legitimate agents reach the same random scalar value $z$ almost surely. However, the consensus value $z$ can be outside the convex hull of the initial values $x_{\ccalL}(0)$ of legitimate agents, unlike the result of the standard consensus process. 
\subsection{Deviation from Consensus under Intermittent Failures and Attacks}

In this part, we identify the deviation from the nominal consensus value $\mathbf{1}\nu^T x_{\ccalL}(0)$, which is a desired outcome for a system with malicious agents as $\lim_{t\xrightarrow{} \infty} \barwl^t=\mathbf{1}\nu^T$, while in the ideal case we would have $\wl(t)=\barwl$ and $\wm(t)=\mathbf{0}$ for all $t \ge T_0-1$. Since we do not have any structural assumptions on the dynamics of malicious agents in Eq.~\eqref{eq_con_dy}, we adopt the worst-case approach, based on the elimination of weights of malicious agents in a (random) finite time.
%Throughout this section, the detailed proofs are delegated to the extended version of this study \cite{aydin2024resilient}.

We first quantify the probability that the system does not follow the nominal consensus dynamics after the observation window $T_0$. 
\vspace{-2.5 mm}
\begin{lemma} \label{lem_WlnotbarW}
 Suppose Assumption \ref{as_trust} holds. Let $\xi>0$ and $\gamma \in (0.5,1)$ be parameters as defined in Algorithm \ref{alg_trust}. For $T_0 > \left(\frac{\xi}{\lambda}\right)^{1/(1-\gamma)}$, the probability of the event that there exists a time instance $k\ge T_0-1$, such that the matrix $\wl(k)$ at time $k$ is not equal to the nominal weight matrix $\widebar{W}_{\ccalL}$, is bounded as follows:
\begin{align*}
    &\mathbb{P} ( \exists k \ge T_0-1: \wl(k) \neq \widebar{W}_{\ccalL} ) \nonumber\\
    &\le |\mathcal{L}|.|\mathcal{M}|(|\mathcal{L}|+1)\Bigg( \frac{\exp((-\xi^2/2) \lambda^2 T_0)}{1-\exp((-\xi^2/2) \lambda^2)}\Bigg) \nonumber \\
    &\hspace{-3 mm}+ |\mathcal{L}|^3 \Bigg(\frac{2^{\frac{1}{2\gamma-1}}}{(2\gamma-1) \xi^{\frac{2}{2\gamma-1}}}\Bigg) \Gamma \Bigg(\frac{1}{(2\gamma-1)}, (\xi^2/2) (T_0-1)^{2\gamma-1} \Bigg),
    %&\le |\mathcal{L}| (|\mathcal{L}| .|\mathcal{N}|+|\mathcal{M}|) \Bigg(\frac{2^{\frac{1}{2\gamma-1}}}{(2\gamma-1) \xi^{\frac{2}{2\gamma-1}}}\Bigg) \Gamma \Bigg(\frac{1}{(2\gamma-1)}, (\xi^2/2) (T_0-1)^{2\gamma-1} \Bigg).
\end{align*}
where $\lambda = \min_{m \in \mathcal{M}} (d-c_m)$ and 
$\Gamma(.,.)$ is the upper incomplete gamma function i.e.,
$\Gamma(s,q)=  \int_{q}^\infty \exp(-u)  u^{s-1} du.$
\end{lemma}
\begin{proof}
    Firstly, we rewrite the event as a union of the following events: 
    \begin{align*}
    &\mathbb{P} ( \exists k \ge T_0-1: \wl(k) \neq \widebar{W}_{\ccalL} ) \\
    &= \mathbb{P} \Big ( \bigcup_{k \ge T_0-1} \{ \wl(k) \neq \widebar{W}_{\ccalL} \} \Big ) \\
    &\le  \sum_{k \ge T_0-1}  \mathbb{P} (\{ \wl(k) \neq \widebar{W}_{\ccalL} \} ), 
    \end{align*}
where the last inequality follows from the union bound.
    The event $\{ \wl(k) \neq \widebar{W}_{\ccalL} \} $ is equivalent to having at least one misclassified agent among neighbors of a legitimate agent at time $k$. Hence,

\begin{align}
    &\mathbb{P} ( \wl(k) \neq \widebar{W}_{\ccalL} ) \cr
    &=   \mathbb{P} \Big( \bigcup_{\substack{i\in \mathcal{L} \\  l \in \Nil}} \{ l \not \in \Ni (k) \} \bigcup_{\substack{i\in \mathcal{L} \\  m \in \Nim}} \{ m  \in \Ni (k) \} \Big) \cr
    &\le  \sum_{\substack{i\in \mathcal{L} \\  l \in \Nil}}  \mathbb{P} ( l \not \in \Ni (k))  + \sum_{\substack{i\in \mathcal{L} \\  m \in \Nim}} \mathbb{P}( m  \in \Ni (k)) \cr
    &\le |\mathcal{L}|.|\mathcal{M}|(|\mathcal{L}|+1)(\exp((-\xi^2/2) \lambda^2(k+1)))) \cr
    &+|\mathcal{L}|^3\exp((-\xi^2/2) (k+1)^{(2\gamma-1)}) \label{eq_sum_bound},
\end{align}
where the last inequality follows from Lemmas \ref{lem_misp_legit}-\ref{lem_misp_mal}. As both terms in Eq. \eqref{eq_sum_bound} are geometrically decaying given $\lambda >0$ and  $\gamma \in (0.5, 1) $, the joint sum over time and agents is finite and can be estimated as follows:
%\oea{We have an issue here. I don't think the infinite sum here has a closed-form solution.}
\begin{align*}
    &\mathbb{P} ( \exists k \ge T_0-1, \exists i \in \mathcal{L}: \wl(k) \neq \widebar{W}_{\ccalL} ) \\
    &\le   \sum_{k = T_0-1}^\infty  \mathbb{P} (\{ \wl(k) \neq \widebar{W}_{\ccalL} \} ) \\
    &\le |\mathcal{L}|.|\mathcal{M}|(|\mathcal{L}|+1)\sum_{k=T_0-1}^\infty\exp((-\xi^2/2) (d-c_m)^2(k+1)) \nonumber\\
    &+|\mathcal{L}|^3 \int_{T_0-2}^\infty \exp (-(\xi^2/2)(s+1)^{(2\gamma -1)}) ds\cr
%\end{align*}    
%    where we bound the first term with the given sum, and we also use integral bound for the second part. Then, we further have
%\begin{align}
    & \le |\mathcal{L}|.|\mathcal{M}|(|\mathcal{L}|+1)\Bigg( \frac{\exp((-\xi^2/2) \lambda^2 T_0)}{1-\exp((-\xi^2/2) \lambda^2)}\Bigg) \\
    &+ |\mathcal{L}|^3 \Bigg(\frac{2^{\frac{1}{2\gamma-1}}}{(2\gamma-1) \xi^{\frac{2}{2\gamma-1}}}\Bigg) \Gamma \Bigg(\frac{1}{(2\gamma-1)}, (\xi^2/2) (T_0-1)^{2\gamma-1} \Bigg),
\end{align*}
where the first term in the last inequality follows from the geometric sum in the second inequality. The second term 
in the last inequality is obtained by using a change of variable in the integral in the second inequality and 
the definition of the upper incomplete gamma function, i.e., $\Gamma(s,q)=  \int_{q}^\infty \exp(-u)  u^{s-1} du.$
\end{proof}

In Lemma \ref{lem_WlnotbarW}, we derived the bound, using the fact that not following the nominal consensus dynamics is equivalent to a misclassification error done by at least one legitimate agent in the system. Now, we are ready to analyze the deviation resulting from the misclassification of legitimate agents. 
\vspace{-1 mm}
\begin{lemma} \label{lem_dev_p1}
    Suppose Assumption \ref{as_trust} holds. Let $\xi>0$ and $\gamma \in (0.5,1)$ be as defined in Algorithm \ref{alg_trust}. Let $\varphi_i(T_0,t)$ be a deviation experienced by a legitimate agent $i \in \mathcal{L}$ defined formally as follows, for all $t\ge T_0$,

    \begin{equation}\label{eq_dev}
        \varphi_i(T_0,t):= \Bigg| \Bigg [ \Tilde{x}_{\ccalL}(T_0,t)-\bigg ( \prod_{k=T_0-1}^{t-1} \barwl \bigg)\xl(0) \Bigg]_i \Bigg|.
    \end{equation}
Then, for an error level $\delta>0$ and $T_0 > \left(\frac{\xi}{\lambda}\right)^{1/(1-\gamma)}$, we have
\begin{align*}
    \mathbb{P} \Big ( \max_{i \in \mathcal{L}} \: \limsup_{t \rightarrow \infty} \varphi_i(T_0,t) > \frac{2\eta}
    {\delta} g_{\mathcal{L}} (T_0) \Big ) < \delta , 
\end{align*}
where $\eta \ge \sup_{i \in \mathcal{N}, t \in \mathbb{N} } |x_i(t)|$, $\lambda = \min_{m \in \mathcal{M}} (d-c_m)$, and we define 
%{\bf \color{red} $g_{\mathcal{L}} (T_0,t)$ is above, while below not}
\begin{align} \label{eq_gl}
 &g_{\mathcal{L}} (T_0):= |\mathcal{L}|.|\mathcal{M}|(|\mathcal{L}|+1)\Bigg( \frac{\exp((-\xi^2/2) \lambda^2 T_0)}{1-\exp((-\xi^2/2) \lambda^2)}\Bigg) \nonumber\\
    &\hspace{-3 mm}+ |\mathcal{L}|^3 \Bigg(\frac{2^{\frac{1}{2\gamma-1}}}{(2\gamma-1) \xi^{\frac{2}{2\gamma-1}}}\Bigg) \Gamma \Bigg(\frac{1}{(2\gamma-1)}, (\xi^2/2) (T_0-1)^{2\gamma-1} \Bigg). 
\end{align}
\end{lemma}

\begin{proof}
    We rewrite the deviation $\varphi_i(T_0,t)$ as per the definition in \eqref{eq_con_part1}, as follows:
    \begin{equation*}
        \varphi_i(T_0,t)= \Bigg| \Bigg [  \bigg (\bigg( \prod_{k=T_0-1}^{t-1} \wl(k) \bigg)-\barwl^{t-T_0} \bigg)\xl(0) \Bigg]_i \Bigg|.
    \end{equation*}
We define the random variable $T_f(T_0,t)$, as follows:
\begin{equation}\label{eq_Tf_T0}
   T_f(T_0,t)=  \begin{cases*}
0 \: \parbox[t]{8.5cm}{\text{if  $\wl(k)= \barwl$ 
 for $T_0-1\le k\le t-1$}},\\
\underset{\substack{ s.t. \:  \wl(k+T_0-1) = \barwl\\ k \in \{0,\cdots, t- T_0 \}}} {\sup{k+1}}  \text{ otherwise}.
\end{cases*}
\end{equation}
We further specify the difference between the products of true and assigned matrices with respect to the random variable $T_f(T_0,t)$ and obtain
\begin{equation*}
    \Delta(\wl,T_f) =  \bigg ( \prod_{k=T_0-1}^{T_0+T_f(T_0,t)-2} \wl(k) \bigg) -\bigg ( \prod_{k=T_0-1}^{T_0+T_f(T_0,t)-2} \barwl \bigg).
\end{equation*}
Using the difference $\Delta(\wl,T_f)$  we bound the deviation as follows:
\begin{align*}
 \varphi_i(T_0,t)&= \Bigg| \Bigg [  \bigg ( \prod_{k=T_0+T_f(T_0,t)-1}^{t-1} \barwl \bigg) \Delta(\wl,T_f) \xl(0) \Bigg]_i  \Bigg|   \\
 &\le \max_{i \in \mathcal{L}} |[\Delta(\wl,T_f) \xl(0)]_i |\\
 &\le \eta \max_{i \in \mathcal{L}} [|\Delta(\wl,T_f) | \mathbf{1} ]_i,
\end{align*}
where we state the upper bound using the fact that $\barwl$ is a row-stochastic matrix, and $|\xl(0)|_i \le \eta$ for any agent $i \in \mathcal{L}$. For a row-stochastic matrix $\barwl$
and a row-substochastic matrix $\wl(k)$ at any time $k$, it holds $[\barwl]_{ii} \ge 1/ n_w $ and $[\wl(k)]_{ii} \ge 1/ n_w$ by their definitions in Eqs.~\eqref{eq_wij}-\eqref{eq_barwij}, where $n_w= \max \{ |\mathcal{L}| +|\mathcal{M}|, \kappa\}$. Hence, by Lemma 4 of \cite{yemini2021characterizing} we have the following relation
\begin{align*}
    [|\Delta(\wl,T_f) | \mathbf{1} ]_i \le 2 \bigg [ 1- \bigg( \frac{1}{n_w}\bigg )^{T_f(T_0,t)}\bigg],
\end{align*}
 which implies that  $\varphi_i(T_0,t) \le 2 \eta [ 1- ( 1/n_w)^{T_f(T_0,t)} ] $  for any $t \ge T_0$ and $i \in \mathcal{L}$. Hence,
 \begin{equation} \label{eq_Markov_varphi}
     \max_{i \in \mathcal{L}} \: \limsup_{t \to \infty} \varphi_i(T_0,t) \le \lim_{t \to\infty} 2 \eta \bigg [ 1- \bigg( \frac{1}{n_w}\bigg )^{T_f(T_0,t)}\bigg].
 \end{equation}
By defining $\widebar{\varphi} (T_0,t) := 2 \eta \bigg [ 1- \bigg( \frac{1}{n_w}\bigg )^{T_f(T_0,t)}\bigg]$ and with the  result $T_f(T_0,t)=t$ almost surely for $t \ge T_f$ by Lemma~\ref{lem_als_w}, Markov's inequality gives the following bound,
%{\bf\color{red} Red zone $g$}
\begin{equation}\label{eq-barphi}
   \mathbb{P} ( \lim_{t \rightarrow \infty} \widebar{\varphi} (T_0,t) \le \frac{2\eta}
    {\delta} g_{\mathcal{L}} (T_0,t)  ) \le \frac{\mathbb{E}  (\lim_{t \rightarrow \infty} \widebar{\varphi} (T_0,t))} {\frac{2\eta}
    {\delta} g_{\mathcal{L}} (T_0)}.
\end{equation}
The sequence of the random variables $\{\widebar{\varphi} (T_0,t)\}_{t \ge T_0}$  is nonnegative and nondecreasing, as we have $0 \le \widebar{\varphi} (T_0,t) \le \widebar{\varphi} (T_0,t+1) \le 2 \eta $ and $0 \le T_f(T_0,t) \le  T_f(T_0,t+1)$ by the definition in Eq \eqref{eq_Tf_T0}. the Monotone Convergence Theorem allows us to exchange the order of the limit and the expectation operator. Therefore, 
\begin{align*}
    \mathbb{E}  (\lim_{t \rightarrow \infty} \widebar{\varphi} (T_0,t)) &= \lim_{t \rightarrow \infty} \mathbb{E}  (\widebar{\varphi} (T_0,t))\\
    &= 2\eta  \Bigg [ 1- \lim_{t \rightarrow \infty} \mathbb{E} \Bigg (\bigg( \frac{1}{n_w}\bigg )^{T_f(T_0,t)} \Bigg) \Bigg ] \\
    &\le 2\eta [1- \lim_{t \rightarrow \infty}\mathbb{P} (T_f(T_0,t)=0)] \\
    &= 2 \eta \lim_{t \rightarrow \infty} [1-\mathbb{P} (T_f(T_0,t)=0) ].
\end{align*}
Since $1-\mathbb{P} (T_f(T_0,t)=0) = \mathbb{P} (T_f(T_0,t)>0) $, we have that
\begin{align*}
    \mathbb{E}  (\lim_{t \rightarrow \infty} \widebar{\varphi} (T_0,t)) &\le  2 \eta \lim_{t \rightarrow \infty} \mathbb{P} (T_f(T_0,t)>0) \\
    &\le 2 \eta \: \mathbb{P} ( \exists k \ge T_0-1: \wl(k) \neq \widebar{W}_{\ccalL} ) \\
    &\le 2\eta g_{\mathcal{L}}(T_0).
\end{align*}
where the last inequality follows by Lemma~\ref{lem_WlnotbarW} and the definition of $g_{\mathcal{L}}(T_0)$ in Eq \eqref{eq_gl}. 
The result follows from the preceding relation and Eqs.\eqref{eq_Markov_varphi}--\eqref{eq-barphi}.
 \end{proof}
The result of Lemma \ref{lem_dev_p1} is a direct consequence of the probability of the event we defined in Lemma \ref{lem_WlnotbarW} since the differences between the products of the nominal and the actual matrices are bounded with the given probability. Next, we analyze the remaining part of the deviation resulting from the misclassification of malicious agents, by defining the following quantity for each $i\in\ccalL$,
\begin{equation}\label{eq-mal-infl}
\phi_i(T_0,t)=\eta \sum_{k=T_0-1}^{t-1} \sum_{ j \in \Nim \cap \Ni (t)} \Bigg [ \Bigg ( \prod_{\ell=k+1}^{t-1}   W_{\mathcal{L}} (\ell) \Bigg ) W_{\mathcal{M}} (k) \Bigg]_{ij}.
\end{equation}
The value $\phi_i(T_0,t)$ provides  an upper bound on the elements of the vector $\phi_{\mathcal{M}}(T_0,t)$ of influence of the malicious agents as defined in~\eqref{eq_sep_Dy}, for which we have
\begin{equation*}
|[\phi_{\mathcal{M}}(T_0,t)]_i| \le \max_{i \in \mathcal{L}} \phi_i(T_0,t).
\end{equation*}

\begin{lemma} \label{lem_dev_p2}
 Suppose Assumption \ref{as_trust} holds. Let $\xi>0$ and $\gamma \in (0.5,1)$ be the parameters of Algorithm~\ref{alg_trust}. For an error level $\delta > 0$ and $T_0 > \left(\frac{\xi}{\lambda}\right)^{1/(1-\gamma)}$, %\oea{OEA: Should we add the definition of $\lambda$ here since it is not a part of the algorithm?} %\sa{It is defined in the sentence starting with "where.." , do you want me to move it to here?} we have the following relation \oea{OEA: I think we should move it here since we use it for the first time here. Also, do we need to give the definition of $\eta$ again? We reference it multiple times during the paper and introduce the notation in our problem formulation.}
 %{\bf\color{red} RED ZONE}
\begin{equation*}
\mathbb{P} \Bigg ( \max_{i \in \mathcal{L}} \limsup_{t \to\infty } \phi_i(T_0,t) > \frac{\eta }{ \kappa \delta} g_{\mathcal{M}}(T_0) \Bigg) < \delta
\end{equation*} 
where $\eta \ge \sup_{i \in \mathcal{N}, t \in \mathbb{N} } |x_i(t)|$, $\lambda = \min_{m \in \mathcal{M}} (d-c_m)$, and
\begin{equation} \label{eq_gm}
    g_{\mathcal{M}} (T_0)= \frac{|\mathcal{L}| |\mathcal{M}|\exp((-\xi^2/2) \lambda^2 T_0)}{1-\exp((-\xi^2/2) \lambda^2)}.
\end{equation} 
\end{lemma}
\begin{proof}
Since
\begin{align*}
&\mathbb{P} \Bigg ( \max_{i \in \mathcal{L}} \limsup_{t \to \infty } \phi_i(T_0,t) >\frac{\eta}{ \kappa \delta} g_{\mathcal{M}}(T_0) \Bigg )\\
&= \mathbb{P} \Bigg ( \bigcup_{i \in \mathcal{L}} \limsup_{t \to\infty } \phi_i(T_0,t) > \frac{\eta}{ \kappa \delta} g_{\mathcal{M}}(T_0) \Bigg ),
\end{align*}
using the union bound and Markov's inequality, we obtain
\begin{align*}
&\mathbb{P} ( \max_{i \in \mathcal{L}} \limsup_{t \to \infty } \phi_i(T_0,t) > \frac{\eta}{ \kappa \delta} g_{\mathcal{M}}(T_0))\\
&\le\sum_{ i\in \mathcal{L}} \mathbb{P} ( \limsup_{t \to\infty } \phi_i(T_0,t) > \frac{\eta }{ \kappa \delta} g_{\mathcal{M}}(T_0)) \\
& \le\frac{ \delta\kappa |\mathcal{L}|\,\mathbb{E}(\limsup_{t \to \infty } \phi_i(T_0,t)) }{\eta g_{\mathcal{M}}(T_0)}.
\end{align*}
In the remaining part, we focus on deriving an upper bound for the expectation $\mathbb{E}(\limsup_{t \to\infty } \phi_i(T_0,t))$ by firstly bounding the random variable $\phi_i(T_0,t)$ defined in~\eqref{eq-mal-infl}), as follows:
\begin{align*}
\phi_i(T_0,t)
%&=\eta \sum_{k=T_0-1}^{t-1} \sum_{ j \in \Nim \cap \Ni (t)} \Bigg [ \Bigg ( \prod_{\ell=k+1}^{t-1}   W_{\mathcal{L}} (\ell) \Bigg ) W_{\mathcal{M}} (k) \Bigg]_{ij} \\
%&\sa{\le \eta \sum_{k=T_0-1}^{t-1} \sum_{ j \in \Nim \cap \Ni (t)} \Bigg [W_{\mathcal{M}} (k) \Bigg]_{ij} \: check the error} \\
&\le \eta \sum_{k=T_0-1}^{t-1} \sum_{ j \in \Nim \cap \Ni (t)} \frac{1}{\kappa} \Bigg(\sum_{n\in \ccalL} \tilde{w}_{in} \Bigg)%\label{eq_bound_prod},
\end{align*}
where $\tilde{w}_{in}=[\prod_{l=k+1}^{t-1}   W_{\mathcal{L}} (\ell)]_{in}$ for $(i,n) \in \ccalL \times \ccalL$, and we used the fact that $[W_{\mathcal{M}} (k) ]_{ij} \le 1/\kappa$ for any $(i,j) \in \ccalL \times \ccalM$. As the product of row-(sub)stochactic matrices is still row-(sub)stochactic, it holds $\sum_{n\in \ccalL} \tilde{w}_{in} \le 1$. 
Using the indicator variable $\mathbbm{1}_{\{j \in \Ni (t)\}}$, which equals $1$ if a malicious agent $j$ is a part of a trusted neighborhood and equal to 0 otherwise, we have the following upper-bound on $\bar{\phi}_i (T_0,t)$:
\begin{align*}
\phi_i(T_0,t)&
%\le \eta \sum_{k=T_0-1}^{t-1} \sum_{ j \in \Nim \cap \Ni (t)} \frac{1}{\kappa} \\
&\le  \frac{\eta}{\kappa} \sum_{k=T_0-1}^{t-1} \sum_{ j \in \Nim} \mathbbm{1}_{\{j \in \Ni (t)\}} 
=\bar{\phi}_i (T_0,t).
\end{align*}
The upper bound also holds for the limit superior of both sequences in expectation, i.e.,
\begin{equation*}
\mathbb{E} (\limsup_{t \to\infty } \phi_i(T_0,t)) \le \mathbb{E} (\limsup_{t \to\infty } \bar{\phi}_i(T_0,t)).
\end{equation*}
Since the sequence $\{\bar{\phi}_i (T_0,t)\}_{t\ge T_0}$ is nonnegative and nondecreasing as $t$ increases, we apply the Monotone Convergence Theorem, and obtain
\begin{align*}
\mathbb{E} (\limsup_{t \to\infty } \bar{\phi}_i(T_0,t)) &= \mathbb{E} (\lim_{t \to \infty } \bar{\phi}_i(T_0,t)) \\
&= \lim_{t \to\infty }\mathbb{E} ( \bar{\phi}_i(T_0,t)).
\end{align*}
Next, we analyze the expectation using the linearity of expectation and the fact that the expectation of an indicator variable is equal to the probability of the given event. Thus, we have
\begin{align*}
\lim_{t \to\infty } \mathbb{E}(  \bar{\phi}_i(T_0,t))
&=\frac{\eta}{\kappa}
\lim_{t \to \infty } \mathbb{E}\left(
\sum_{k=T_0-1}^{t-1} \sum_{ j \in \Nim} \mathbbm{1}_{\{j \in \Ni (t)\}}\right)\\
&=\frac{\eta}{\kappa}
\lim_{t \to \infty } \sum_{k=T_0-1}^{t-1} \sum_{ j \in \Nim} \mathbb{P}(j \in \Ni (t)).
\end{align*} 
Using the upper-bound for the probability of misclassification of malicious agents of Lemma~\ref{lem_misp_mal}, we obtain
\begin{align*}
\mathbb{E}  (\limsup_{t \to\infty } &\phi_i(T_0,t)) \le 
\frac{\eta}{\kappa} 
\lim_{t \to \infty } \sum_{k=T_0-1}^{t-1} \sum_{ j \in \Nim} \mathbb{P}(j \in \Ni (t))  \\
&\le \frac{\eta |\mathcal{M}|}{\kappa} 
\lim_{t \to \infty } \sum_{k=T_0-1}^{t-1} \exp(-(\xi^2/2)\lambda^2 (k+1)) \\
&\le \frac{\eta |\mathcal{M}|}{\kappa}  \Bigg( \frac{\exp((-\xi^2/2) \lambda^2 T_0)}{1-\exp((-\xi^2/2) \lambda^2)}\Bigg).
\end{align*} 
Thus, for a given error tolerance $\delta >0$, we obtain
\begin{align*}
&\mathbb{P} ( \max_{i \in \mathcal{L}} \limsup_{t \to \infty } \phi_i(T_0,t) >\frac{\eta }{ \kappa \delta} g_{\mathcal{M}}(T_0)) \\
& \le\frac{\delta\kappa|\mathcal{L}|\,\mathbb{E}(\limsup_{t \to \infty } \phi_i(T_0,t)) }{\eta g_{\mathcal{M}}(T_0)}\le \delta.
\end{align*}
\end{proof}

Similar to, the result of Lemma \ref{lem_dev_p1}, we derived the upper bound using misclassification probabilities in Lemma \ref{lem_dev_p2}. This time, the result specifically follows from misclassification probabilities of malicious agents. Finally, we state the main result on the deviation from the nominal consensus process.

\begin{theorem} \label{thm_dev}
Suppose Assumption \ref{as_trust} holds. Let $\xi>0$ and $\gamma \in (0.5,1)$ be as given in Algorithm \ref{alg_trust}. For an error level $\delta >0$ and $T_0 > \left(\frac{\xi}{\lambda}\right)^{1/(1-\gamma)}$,
we have
\begin{align*}
    \mathbb{P}( \max \: \limsup_{t \xrightarrow[]{} \infty} |[ \xl(T_0,t)- \mathbf{1} \nu^T \xl(0) ]_i | &< \Delta_{\max} (T_0,\delta) ) \cr
    \ge 1-\delta,
\end{align*}
where $\Delta_{\max} (T_0,\delta)= 2(\frac{2\eta} {\delta} g_{\mathcal{L}}  (T_0)+ \frac{\eta} {\kappa \delta} g_{\mathcal{M}}(T_0))$.
\end{theorem}

\begin{proof}
By the triangle inequality we have
    \begin{align*}
       &|[ \xl(T_0,t)- \mathbf{1} \nu^T \xl(0) ]_i| \nonumber \\
       &\le \Bigg | \Bigg [\xl(T_0,t) - \Bigg ( \prod_{k=T_0}^{t-1} \barwl \Bigg ) \xl(0) \Bigg]_i \Bigg| \nonumber\\
       &\hspace{0.3 cm}+ \Bigg | \Bigg [  \Bigg ( \prod_{k=T_0}^{t-1} \barwl \Bigg ) \xl(0)- \mathbf{1} \nu^T \xl(0) \Bigg]_i \Bigg|.
    \end{align*}
Given that  
$\lim_{t\to\infty}\barwl^t={\bf 1}\nu^T$ almost surely for a stochastic vector $\nu>{\bf 0}$, we have that
\begin{align*}
    \lim_{t \xrightarrow{} \infty} \Bigg | \Bigg [  \Bigg ( \prod_{k=T_0}^{t-1} \barwl \Bigg ) \xl(0)- \mathbf{1} \nu^T \xl(0) \Bigg]_i \Bigg| =0.
\end{align*}
This also implies the equivalence between the following events
\begin{align}
    &\{ \max_{i \in \ccalL} \: \limsup_{t \xrightarrow[]{} \infty} |[ \xl(T_0,t)- \mathbf{1} \nu^T \xl(0) ]_i | \ge \Delta_{\max} (T_0,t))  \nonumber \} \\
    =&  \Bigg \{ \max_{i \in \ccalL} \:
    \limsup_{t \xrightarrow{} \infty}
    \Bigg| \Bigg [\xl(T_0,t) - \Bigg ( \prod_{k=T_0}^{t-1} \barwl \Bigg ) \xl(0) \Bigg]_i \Bigg| \nonumber \\
    & \hspace{5.5 cm } \ge \Delta_{\max}(T_0,t) \Big\}.
\end{align}
Then, using the triangle inequality and the relation $\xl (T_0,t)= \Txl  (T_0,t)+ \phi_{\mathcal{M}} (T_0,t)$ defined in Eq. \eqref{eq_con_sum}, it follows that 

\begin{align}
    &\Bigg| \Bigg [\xl(T_0,t) - \Bigg ( \prod_{k=T_0}^{t-1} \barwl \Bigg ) \xl(0) \Bigg]_i \Bigg|  \nonumber \\
    &\le \Bigg| \Bigg [\Txl(T_0,t) - \Bigg ( \prod_{k=T_0}^{t-1} \barwl \Bigg ) \xl(0) \Bigg]_i \Bigg|  + | [\phi_{\mathcal{M}} (T_0,t)]_i|.
\end{align}
Next,  to analyze the probability of maximal deviation from the nominal consensus, we define the following events
%{\bf \color{red} RED ZONE}
\begin{align}
    &E_1=\Bigg \{ \max_{i \in \ccalL} \: \limsup_{t \xrightarrow[]{} \infty} \Bigg| \Bigg [\xl(T_0,t) - \Bigg ( \prod_{k=T_0}^{t-1} \barwl \Bigg ) \xl(0) \Bigg]_i \Bigg| \nonumber\\
    &\hspace{5 cm} \ge \Delta_{\max} (T_0,t)\Bigg \}  \\
    &E_2= \Bigg \{ \max_{i \in \ccalL} \: \limsup_{t \xrightarrow[]{} \infty} \Bigg| \Bigg [\Txl(T_0,t) - \Bigg ( \prod_{k=T_0}^{t-1} \barwl \Bigg ) \xl(0) \Bigg]_i \Bigg| \nonumber \\
    &\hspace{5 cm}\ge  \frac{4\eta} {\delta} g_{\mathcal{L}}(T_0) \Bigg \} \\
    &E_3=  \{ \max_{i \in \ccalL} \: \limsup_{t \xrightarrow[]{} \infty} | [\phi_{\mathcal{M}} (T_0,t)]_i| \ge  \frac{\eta} {\kappa \delta} 2g_{\mathcal{M}}(T_0))\}
   %2 \frac{2}{\eta}{\delta} g_{\mathcal{L}\}.
\end{align}
As the complements of the events satisfy the relation $(E_2 \cup E_3)^C= E_2^C \cap E_3^C \subseteq E_1^C$, we bound the probability of the event that the maximum deviation is at least $\Delta_{\max}(T_0,t)$ by the union bound, and obtain
\begin{equation}
   \mathbb{P}(E_1) \le \mathbb{P}(E_2)+\mathbb{P}(E_3) .   
\end{equation}
\begin{comment}
\begin{align}
    &\mathbb{P} \Bigg( \max \: \limsup_{t \xrightarrow[]{} \infty}\Bigg| \Bigg [\xl(T_0,t) - \Bigg ( \prod_{k=T_0}^{t-1} \barwl \Bigg ) \xl(0) \Bigg]_i \Bigg| \ge  \Delta_{\max}(T_0,t) \Bigg ) \\
    &\le \mathbb{P} \Bigg( \max \: \limsup_{t \xrightarrow[]{} \infty} \Bigg| \Bigg [\Txl(T_0,t) - \Bigg ( \prod_{k=T_0}^{t-1} \barwl \Bigg ) \xl(0) \Bigg]_i \Bigg| \ge  2\Tilde{g}_{\mathcal{L}}(T_0,t) \Bigg) \\
    &+ \mathbb{P} ( \max \: \limsup_{t \xrightarrow[]{} \infty}| [\phi_{\mathcal{M}} (T_0,t)]_i| \ge 2\Tilde{g}_{\mathcal{M}}(T_0,t)) .
\end{align}{\color{orange}
\end{comment}
 Further, we have $2\frac{2\eta} {\delta} g_{\mathcal{L}}(T_0) =\frac{2\eta} {\delta/2} g_{\mathcal{L}}(T_0)$ and $2\frac{\eta} {\kappa \delta} g_{\mathcal{M}}(T_0)=\frac{\eta} {\kappa \delta/2} g_{\mathcal{M}}(T_0)$.  Therefore, $P(E_2) < \delta/2$ and $P(E_3) < \delta/2$, implying that 
%{\bf \color{red} RED ZONE} by their definitions in Eqs.~\eqref{eq_delta_max}{\color{orange}
\begin{align*}
   &\mathbb{P}( \max_{i \in \ccalL} \: \limsup_{t \xrightarrow[]{} \infty} |[ \xl(T_0,t)- \mathbf{1} \nu^T \xl(0) ]_i | < \Delta_{\max} (T_0,\delta) )\\
   &=1-\mathbb{P}(E_1)\ge  1-(\delta/2+\delta/2)=1-\delta.
\end{align*}
\end{proof}
We examined the deviation from the nominal consensus process. The result follows from the bounds we derived for each part of the deviation (Lemmas \ref{lem_dev_p1}-\ref{lem_dev_p2}). The result shows that as agents wait for more to start the consensus process, \textit{i.e.} with increasing $T_0$, they have tighter bounds on the probabilities. 

\section{Numerical Studies}
In this section, we assess the performance of our proposed algorithm against different type of malicious attacks in numerical studies. We consider a system with $10$ legitimate and $15$ malicious agents. We choose the number of malicious agents higher than the number of legitimate agents to consider a challenging scenario. We construct the communication graph as follows: first, we generate a cycle graph among the legitimate agents and then add $10$ more random edges between them. The malicious agents form random connections to every other agent with probability $0.2$ while we also ensure that they are connected to at least one legitimate agent. We sample agents' initial values from the uniform distribution $\mathcal{U}[-4,4]$ once for both legitimate and malicious agents. Legitimate agents follow the consensus dynamics given in \cref{eq_con} with $\kappa=10.$ The legitimate neighbors' trust values are sampled from the uniform distribution $\mathcal{U}[0.3,1]$ resulting in the expected value $d=0.65$. To model the case where each malicious agent has a different expected value, we choose their expectations from the uniform distribution $\mathcal{U}[0,0.45]$.

\textbf{Attack Models:} We consider two types of attack in our experiments: 1) Consistent attacks where malicious agents always send $\eta$ (or $-\eta$) if the true consensus value is negative (positive). During each communication, a malicious agent $m$'s trust value is sampled from the uniform distribution $\mathcal{U}[2c_m-1,1]$ with probability $p_m$ and from $\mathcal{U}[0.3,1]$ (the same distribution as the legitimate neighbors) with probability $1-p_m$. This attack type corresponds to a worst case scenario that is captured by the analysis where attackers always inject wrong information to the system. 2) Intermittent failures where malicious nodes follow the same consensus update rule as the legitimate agents considering all their neighbors as trustworthy. However, a malicious node $m$ sends $\eta$ (or $-\eta$) to its neighbors with probability $p_m$ instead of its true variable $x_m(t)$ if the true optimal value is negative (positive). When a malicious agent sends $\eta$ to its neighbors, its trust value is sampled from the uniform distribution $\mathcal{U}[2c_m-1,1]$ and from $\mathcal{U}[0.3,1]$ (the same distribution with the legitimate neighbors) when it sends  $x_m(t)$.

For both attacks, we assume that all malicious agents have the same $p_m$, and consider two cases with $p_m=0.2$ and $p_m=0.8.$ We use $\xi = 0.15$ and $\gamma=0.7$ as the parameters of our learning algorithm \cref{alg_trust_neig}. We use $T_0=60$ as it satisfies the largest theoretical lower bound on $T_0$ given in \cref{thm_dev} for all cases. We track the maximum deviation from the nominal consensus value over time. The results are shown in \cref{fig:numerical_studies}. Note that these results are averaged over 100 trials for each setup, where the communication graph, the initial values and the expected values of agents are fixed across the trials. In all cases and trials, we observe that agents reach consensus, as predicted by \cref{cor_con}. Moreover, we can see that the probability of being observable, $p_m$, has the highest impact on the deviation, as it affects the misclassification probabilities (see \cref{lem_misp_mal} and \cref{lem_misp_legit}). As expected, consistent attacks have more impact on the system when the attack probability is low ($p_m=0.2$) since malicious agents are always inserting a constant value $\eta$ (or $-\eta$) to the system, and they stay undetected for a longer time. When the attack probability is high ($p_m=0.8$), malicious agents get detected quickly, and the errors mainly stem from misclassified legitimate agents. This happens because the threshold for adding agents to the trusted neighborhood is narrow due to the window indicated through the parameters $\xi$ and $\gamma$.
\begin{figure}[ht]
    \centering
    \includegraphics[width=0.45\textwidth]{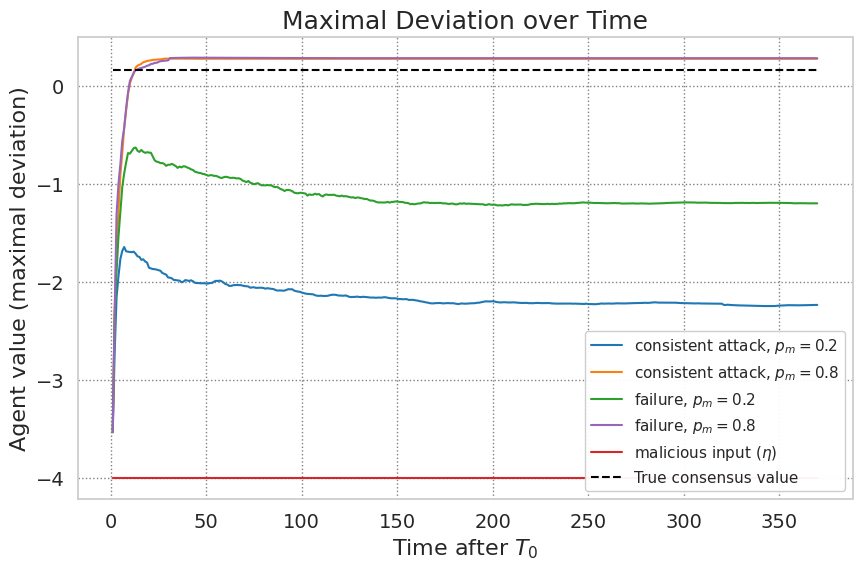}
      \caption{The maximal deviation of the legitimate agents from the nominal consensus value. Malicious input $\eta$ shows the maximum impact that malicious agents can have on the system.
      %\sa{I think it is better to state the main observation in the main body of the text rather than inside of the figure caption}
      }
\label{fig:numerical_studies}
\vspace{-3 mm}
\end{figure}
\vspace{-3 mm}
\section{Conclusion}
In this paper, we studied the multi-agent resilient consensus problem where agents send their values over an undirected and static communication network. Assuming trust observations are available, we considered the scenarios with intermittent faulty or malicious transmissions, where the classification of agents based on constant thresholds may fail. We developed a novel detection algorithm to let legitimate agents determine their neighbors' types correctly. We showed that misclassification probabilities go to $0$ with geometric rates, implying that after some finite and random time, all agents are correctly classified. We also proved that agents reach a consensus almost surely asymptotically. For a pre-determined error tolerance, we derived the maximal deviation from the nominal consensus process, which is a function of the observation window and the number of legitimate and malicious agents, in addition to the parameters of the detection algorithm. Numerical experiments showed the convergence of the consensus process and the deviation under different scenarios. We consider this work as an important step toward developing resilient algorithms against smarter attackers going beyond intermittent attacks and strategically decide when to attack multi-agent systems in the future work

\section{Acknowledgements}
We thank Arif Kerem Dayi and Aron Vekassy for the comments that greatly contributed to the manuscript.

\bibliographystyle{IEEEtran}
%\bibliographyits
\bibliography{bibliography}

\begin{thebibliography}{10}
\providecommand{\url}[1]{#1}
\csname url@samestyle\endcsname
\providecommand{\newblock}{\relax}
\providecommand{\bibinfo}[2]{#2}
\providecommand{\BIBentrySTDinterwordspacing}{\spaceskip=0pt\relax}
\providecommand{\BIBentryALTinterwordstretchfactor}{4}
\providecommand{\BIBentryALTinterwordspacing}{\spaceskip=\fontdimen2\font plus
\BIBentryALTinterwordstretchfactor\fontdimen3\font minus
  \fontdimen4\font\relax}
\providecommand{\BIBforeignlanguage}[2]{{%
\expandafter\ifx\csname l@#1\endcsname\relax
\typeout{** WARNING: IEEEtran.bst: No hyphenation pattern has been}%
\typeout{** loaded for the language `#1'. Using the pattern for}%
\typeout{** the default language instead.}%
\else
\language=\csname l@#1\endcsname
\fi
#2}}
\providecommand{\BIBdecl}{\relax}
\BIBdecl

\bibitem{degroot1974reaching}
M.~H. DeGroot, ``Reaching a consensus,'' \emph{Journal of the American
  Statistical association}, vol.~69, no. 345, pp. 118--121, 1974.

\bibitem{olfati2004consensus}
R.~Olfati-Saber and R.~M. Murray, ``Consensus problems in networks of agents
  with switching topology and time-delays,'' \emph{IEEE Transactions on
  Automatic Control}, vol.~49, no.~9, pp. 1520--1533, 2004.

\bibitem{NedicOzdaglar}
A.~Nedic and A.~Ozdaglar, ``Distributed subgradient methods for multiagent
  optimization,'' \emph{IEEE Trans. Autom. Control}, vol.~54, no.~1, 2009.

\bibitem{nedic2018distributed}
A.~Nedi{\'c} and J.~Liu, ``Distributed optimization for control,'' \emph{Annual
  Review of Control, Robotics, and Autonomous Systems}, vol.~1, pp. 77--103,
  2018.

\bibitem{kia2019tutorial}
S.~S. Kia, B.~Van~Scoy, J.~Cortes, R.~A. Freeman, K.~M. Lynch, and S.~Martinez,
  ``Tutorial on dynamic average consensus: The problem, its applications, and
  the algorithms,'' \emph{IEEE Control Systems Magazine}, vol.~39, no.~3, pp.
  40--72, 2019.

\bibitem{bullo2006rendezvous}
J.~Cortes, S.~Martinez, and F.~Bullo, ``Robust rendezvous for mobile autonomous
  agents via proximity graphs in arbitrary dimensions,'' \emph{IEEE
  Transactions on Automatic Control}, vol.~51, no.~8, pp. 1289--1298, 2006.

\bibitem{martinez2009distributed}
S.~Martinez, ``Distributed interpolation schemes for field estimation by mobile
  sensor networks,'' \emph{IEEE Transactions on Control Systems Technology},
  vol.~18, no.~2, pp. 491--500, 2009.

\bibitem{pasqualetti2011consensus}
F.~Pasqualetti, A.~Bicchi, and F.~Bullo, ``Consensus computation in unreliable
  networks: A system theoretic approach,'' \emph{IEEE Transactions on Automatic
  Control}, vol.~57, no.~1, pp. 90--104, 2011.

\bibitem{sundaram2010distributed}
S.~Sundaram and C.~N. Hadjicostis, ``Distributed function calculation via
  linear iterative strategies in the presence of malicious agents,'' \emph{IEEE
  Transactions on Automatic Control}, vol.~56, no.~7, pp. 1495--1508, 2010.

\bibitem{dolev1982byzantine}
D.~Dolev, ``The byzantine generals strike again,'' \emph{Journal of
  algorithms}, vol.~3, no.~1, pp. 14--30, 1982.

\bibitem{leblanc2013resilient}
H.~J. LeBlanc, H.~Zhang, X.~Koutsoukos, and S.~Sundaram, ``Resilient asymptotic
  consensus in robust networks,'' \emph{IEEE Journal on Selected Areas in
  Communications}, vol.~31, no.~4, pp. 766--781, 2013.

\bibitem{sundaram2019optimization}
S.~Sundaram and B.~Gharesifard, ``Distributed optimization under adversarial
  nodes,'' \emph{IEEE Transactions on Automatic Control}, vol.~64, no.~3, pp.
  1063--1076, 2019.

\bibitem{rawat2011sensing}
A.~S. Rawat, P.~Anand, H.~Chen, and P.~K. Varshney, ``Collaborative spectrum
  sensing in the presence of byzantine attacks in cognitive radio networks,''
  \emph{IEEE Transactions on Signal Processing}, vol.~59, no.~2, pp. 774--786,
  2011.

\bibitem{gil2017guaranteeing}
S.~Gil, S.~Kumar, M.~Mazumder, D.~Katabi, and D.~Rus, ``Guaranteeing
  spoof-resilient multi-robot networks,'' \emph{Autonomous Robots}, vol.~41,
  pp. 1383--1400, 2017.

\bibitem{cavorsi2023ICRA}
M.~Cavorsi, O.~E. Akgün, M.~Yemini, A.~J. Goldsmith, and S.~Gil, ``Exploiting
  trust for resilient hypothesis testing with malicious robots,'' in \emph{2023
  IEEE International Conference on Robotics and Automation (ICRA)}, 2023, pp.
  7663--7669.

\bibitem{Pierson2016}
\BIBentryALTinterwordspacing
A.~Pierson and M.~Schwager, \emph{Adaptive Inter-Robot Trust for Robust
  Multi-Robot Sensor Coverage}.\hskip 1em plus 0.5em minus 0.4em\relax Cham:
  Springer International Publishing, 2016, pp. 167--183. [Online]. Available:
  \url{https://doi.org/10.1007/978-3-319-28872-7_10}
\BIBentrySTDinterwordspacing

\bibitem{xiong2023securearray}
\BIBentryALTinterwordspacing
J.~Xiong and K.~Jamieson, ``Securearray: improving wifi security with
  fine-grained physical-layer information,'' in \emph{Proceedings of the 19th
  Annual International Conference on Mobile Computing \& Networking}, ser.
  MobiCom '13.\hskip 1em plus 0.5em minus 0.4em\relax New York, NY, USA:
  Association for Computing Machinery, 2013, p. 441–452. [Online]. Available:
  \url{https://doi.org/10.1145/2500423.2500444}
\BIBentrySTDinterwordspacing

\bibitem{gil2019consensus}
S.~Gil, C.~Baykal, and D.~Rus, ``Resilient multi-agent consensus using wi-fi
  signals,'' \emph{IEEE Control Systems Letters}, vol.~3, no.~1, pp. 126--131,
  2019.

\bibitem{gil2023physicality}
S.~Gil, M.~Yemini, A.~Chorti, A.~Nedi{\'c}, H.~V. Poor, and A.~J. Goldsmith,
  ``How physicality enables trust: A new era of trust-centered cyberphysical
  systems,'' \emph{arXiv preprint arXiv:2311.07492}, 2023.

\bibitem{yemini2021characterizing}
M.~Yemini, A.~Nedi{\'c}, A.~J. Goldsmith, and S.~Gil, ``Characterizing trust
  and resilience in distributed consensus for cyberphysical systems,''
  \emph{IEEE Transactions on Robotics}, vol.~38, no.~1, pp. 71--91, 2021.

\bibitem{nurellari2018detection}
E.~Nurellari, D.~McLernon, and M.~Ghogho, ``A secure optimum distributed
  detection scheme in under-attack wireless sensor networks,'' \emph{IEEE
  Transactions on Signal and Information Processing over Networks}, vol.~4,
  no.~2, pp. 325--337, 2018.

\bibitem{kailkhura2015detection}
B.~Kailkhura, Y.~S. Han, S.~Brahma, and P.~K. Varshney, ``Asymptotic analysis
  of distributed bayesian detection with byzantine data,'' \emph{IEEE Signal
  Processing Letters}, vol.~22, no.~5, pp. 608--612, 2015.

\bibitem{akgun2023learning}
O.~E. Akgun, A.~K. Dayi, S.~Gil, and A.~Nedich, ``Learning trust over directed
  graphs in multiagent systems,'' in \emph{Learning for Dynamics and Control
  Conference}.\hskip 1em plus 0.5em minus 0.4em\relax PMLR, 2023, pp. 142--154.

\bibitem{kay1998fundamentals}
S.~M. Kay, \emph{Fundamentals of Statistical Signal Processing Volume II
  Detection Theory}.\hskip 1em plus 0.5em minus 0.4em\relax New Jersey:
  Prentice Hall PTR, 1998.

\bibitem{berger2014kolmogorov}
V.~W. Berger and Y.~Zhou, ``Kolmogorov--smirnov test: Overview,'' \emph{Wiley
  statsref: Statistics reference online}, 2014.

\end{thebibliography}

\end{document}